\documentclass[a4paper]{article}
\usepackage{authblk}
\usepackage[utf8]{inputenc}
\usepackage{enumerate}
\usepackage{amsthm}
\usepackage{amsfonts}
\usepackage{amsmath}
\usepackage{amssymb}
\usepackage{tikz}
\usepackage{rotating}
\usepackage{a4wide}
\usepackage{comment}
\usepackage{hyperref}
\usepackage{cleveref}
\usepackage{mathtools}
\usepackage{xspace}
\usepackage{graphicx}
\usepackage{thmtools, thm-restate}

\usepackage{framed}
\usepackage{tabularx}

\newlength{\RoundedBoxWidth}
\newsavebox{\GrayRoundedBox}
\newenvironment{GrayBox}[1]%
   {\setlength{\RoundedBoxWidth}{.93\columnwidth}
    \def\boxheading{#1}
    \begin{lrbox}{\GrayRoundedBox}
       \begin{minipage}{\RoundedBoxWidth}}%
   {   \end{minipage}
    \end{lrbox}
    \begin{center}
    \begin{tikzpicture}%
       \node(Text)[draw=black!20,fill=white,rounded corners,inner sep=2ex,text width=\RoundedBoxWidth]
             {\usebox{\GrayRoundedBox}};
        \coordinate(x) at (current bounding box.north west);
        \node [draw=white,rectangle,inner sep=3pt,anchor=north west,fill=white]
        at ($(x)+(6pt,.75em)$) {\boxheading};
    \end{tikzpicture}
    \end{center}}

\newenvironment{defproblemx}[1]{\noindent\ignorespaces%
                                \FrameSep=6pt%
                                \parindent=0pt%
                \begin{GrayBox}{#1}%
                \begin{tabular*}{\columnwidth}{!{\extracolsep{\fill}}@{\hspace{.1em}} >{\itshape} p{1.5cm} p{0.86\columnwidth} @{}}%
            }{
                \end{tabular*}%
                \end{GrayBox}%
                \ignorespacesafterend
            }

\newcommand{\defProblemTask}[3]{%
  \begin{defproblemx}{#1}
    Input: & #2 \\
    Task: & #3
  \end{defproblemx}
}

\newcommand{\defProblemQuestion}[3]{%
  \begin{defproblemx}{#1}
    Input: & #2 \\
    Question: & #3
  \end{defproblemx}
}

\usetikzlibrary{calc,backgrounds,fit}

\newtheorem{observation}{Observation}
\newtheorem*{notation}{Notation}

\DeclareMathOperator{\pw}{pw}
\DeclareMathOperator*{\argmin}{\arg\!\min}
\DeclareMathOperator{\dist}{dist}
\DeclareMathOperator{\poly}{poly}

\newcommand{\idx}{\text{index}}
\newcommand{\vertex}{\text{vertex}}
\newcommand{\tbd}{m}
\newcommand{\tbdd}{n}
\newcommand{\otherN}{\nu}
\newcommand{\lbc}{\textsc{Length-Bounded Cut}}
\newcommand{\s}[1]{\ensuremath{b_{#1}}}
\newcommand{\e}[1]{\ensuremath{f_{#1}}}

\newcommand{\D}{\ensuremath{\vec{D}}}
\newcommand{\tor}{\ensuremath{\Delta}}

\tikzstyle{vertex}=[draw, circle, fill, inner sep = 2pt]
\newcommand{\tikzDrawIntervall}[5]{
		\node[draw, circle, inner sep=1pt, fill=black,#4, label=left:$#5$] at (#2,#1) (1#5) {};
		\node[draw, circle, inner sep=1pt, fill=black,#4] at (#3,#1) (2#5) {};
		\node[radius=0pt] at (#2/2+#3/2,#1-.1) (3#5) {};
		\draw[very thick,#4] (1#5) edge (2#5);
}

\newcommand{\LCut}{\textsc{Length-Bounded Cut}\xspace}

\newtheorem{theorem}{Theorem}
\theoremstyle{definition}

\newtheorem{definition}{Definition}

\newtheorem{lemma}{Lemma}
\newtheorem{corollary}{Corollary}

\newtheorem{claim}{Claim}

\title{Length-Bounded Cuts: Proper Interval Graphs and Structural Parameters}

\author[1]{Matthias Bentert}
\author[1]{Klaus Heeger}
\author[2]{Du\v{s}an Knop}
\affil[1]{Technische Universität Berlin, Chair of Algorithmics and Computational Complexity, \texttt{\{matthias.bentert,heeger\}@tu-berlin.de}}
\affil[2]{Czech Technical University in Prague, Department of Theoretical Computer Science, \texttt{dusan.knop@fit.cvut.cz}}
\date{}

\begin{document}

\maketitle

\begin{abstract}
In the presented paper we study the \LCut problem for special graph classes as well as from a parameterized-complexity viewpoint.
Here, we are given a graph~$G$, two vertices~$s$ and~$t$, and positive integers~$\beta$ and $\lambda$.
The task is to find a set of edges~$F$ of size at most~$\beta$ such that every $s$-$t$-path of length at most~$\lambda$ in~$G$ contains some edge in~$F$.

Bazgan et al.~\cite{BazganFNNS19} conjectured that \LCut admits a polynomial-time algorithm if the input graph~$G$ is a~proper interval graph.
We confirm this conjecture by showing a dynamic-programming based polynomial-time algorithm.
We strengthen the W[1]-hardness result of Dvořák and Knop~\cite{DvorakK18}.
Our reduction is shorter, seems simpler to describe, and the target of the reduction has stronger structural properties.
Consequently, we give W[1]-hardness for the combined parameter pathwidth and maximum degree of the input graph.
Finally, we prove that \LCut is W[1]-hard for the feedback vertex number.
Both our hardness results complement known XP algorithms.
\end{abstract}

\section{Introduction}

The study of network flows and, in particular, the \textsc{Edge Disjoint Paths} (EDP) problem began in the 1950s by the work of Ford and Fulkerson~\cite{FordFulkerson56} and constitutes a prominent research subarea in graph algorithms since then.
In the EDP problem we are given a graph~$G$, two vertices~$s$ and~$t$, called source and sink, and a positive integer~$\beta$.
The task is then to resolve whether in~$G$ one can find a collection of at least $\beta$ edge-disjoint $s$-$t$-paths.
It is worth pointing out that nowadays there are many more efficient algorithms (then the one of ford and Fulkerson~\cite{FordFulkerson56}) for finding a flow in a given graph (see e.g.~\cite{Dinitz06,MalhotraKM78}).
A natural counterpart of EDP is the \textsc{Edge Cut} problem, where the task is to resolve whether there is a set of (at most) $\beta$ edges $F$ such that there is no $s$-$t$-path in the graph~$G-F$.
There is a strong dual relation ship between EDP and \textsc{Edge Cut} in the sense that if for a given~$\beta$ both problems admit a solution, then the value of~$\beta$ is optimal, that is, it is not possible to find~$\beta+1$ edge disjoint $s$-$t$-paths and removal of any set of $\beta-1$ edges leaves $s$ and $t$ in the same connected component.
Consequently, both problems admit an efficient (polynomial-time) algorithm, since one can also construct the set of cut edges from a maximum flow.
Quite naturally there are many variants of the above described network flow/cut problems such as e.g.\ multicomodity flows/cuts, unsplitable flows and the related cut problem (see e.g.~\cite{schrijver03} for further examples and exact definitions).
Unlike the basic variant of EDG and \textsc{Edge Cut}, it is not always the case that the flow and the cut belong to the same complexity class, as we shall see, \LCut{} is in fact harder than the respective flow problem.

In our paper we continue the study of the so-called \LCut problem, which is the related cut problem to the variant of EDP where an additional bound~$\lambda$ is given and the sought collection of $s$-$t$-paths can only contain paths with at most $\lambda$ edges.
To the best of our knowledge, these problems have been introduced by Adámek and Koubek~\cite{AdamekK71} and the \textsc{Length-bounded Cut} problem is formally defined as follows.
\defProblemTask{\LCut}
{An undirected graph $G = (V, E)$, two vertices $s$, $t$, and two positive integers~$\beta$,~$\lambda$.}
{Decide whether there exists a subset $F\subseteq E$ with $|F|\le \beta$ such that there is no $s$-$t$-path in $G-F$ of length at most $\lambda$.}
If in the above definition one plugs-in $\lambda=|G|$, then one is left with the \textsc{Edge Cut} problem; a  polynomial-time-solvable problem.
However, Baier et al.~\cite{BaierEHKKPSS10} showed that the \LCut problem is NP-hard already for $\lambda=4$.
On the other hand, the related \textsc{Length-Bounded Flow} problem, where we restrict the flow to paths of length at most $\lambda$, can be solved in polynomial time via a~reduction to linear programming~\cite{BaierEHKKPSS10,MahjoubM10,KolmanS06}.
Before we give an overview of our results, we discuss the related work with the focus on parameterized algorithms and algorithms for special graph classes.

\subsection{Related Work}
Note that the result of Baier et al.~\cite{BaierEHKKPSS10} in fact gives para-NP-hardness for \LCut for the parameter~$\lambda$.
Thus, in order to obtain tractability results one has to either consider a different parameterization or combine~$\lambda$ with some other parameter.
The first study of \LCut from the viewpoint of parameterized complexity was done by Golovach and Thilikos~\cite{GolovachT11}.
They showed that \LCut is in FPT for the combined parameter $\beta+\lambda$.
It is worth noting that parameterization by~$\beta$ only leads to para-NP-hardness as well~\cite{LiMS90,TragoudasV96}.
Later, Fluschnik et al.~\cite{FluschnikHNN18} proved that it is unlikely that a polynomial kernel in $\beta+\lambda$ exists.
Dvořák and Knop~\cite{DvorakK18} considered structural parameters for the \LCut problem.
They showed that it is W[1]-hard when parameterized by the pathwidth of the input graph while it is fixed-parameter tractable when parameterized by treedpeth on the input graph.
It is worth pointing out that \LCut is one of just a few problems with such a parameterized dichotomy.
Kolman~\cite{Kolman18} gave an $O\left( \lambda^\tau \cdot |G| \right)$-time algorithm for \LCut, where $\tau$ is the treewidth of~$G$.
Furthermore, \LCut is in FPT for the parameter~$\lambda$ if~$G$ is planar~\cite{Kolman18} (it remains NP-complete even in this case~\cite{FluschnikHNN18,ZschocheFMN18}).
Bazgan et al.~\cite{BazganFNNS19} studied both restrictions on special graph classes as well as structural parameterizations for \LCut.
They provided an XP algorithm for the parameter~$\Delta$, the maximum degree of the input graph~$G$, and an FPT algorithm for the feedback edge number.
Finally, a polynomial-time algorithm is presented for co-graphs while it remains NP-complete even if the input is restricted to bipartite graphs or split graphs.

\subsection{Our Contribution}
In this paper we mainly continue the study of \LCut for special graph classes as well as from a parameterized-complexity viewpoint.
Bazgan et al.~\cite{BazganFNNS19} conjectured that \LCut can be solved in polynomial time on proper interval graphs which we confirm here:
We give a dynamic-programming based algorithm in \Cref{thm:lcutPolyPINT}.

We show that the \LCut problem is W[1]-hard for the feedback vertex number in \Cref{thm:LCutWhardFVS}; thus, closing one of the gaps left by Bazgan et al.~\cite{BazganFNNS19}.
With this we hit yet another ambiguity of the \LCut problem:
It is quite common for graph problems that a single W[1]-hardness reduction allows one to bound both feedback vertex number and treedepth (e.g.~\cite{ChopinNNW14,KnopMT19}).
Furthermore, together with the result of Bazgan et al.~\cite{BazganFNNS19} this yield a structural parameter dichotomy, since they provided an FPT algorithm for the feedback edge number.

Last but not least, we show in \Cref{thm:LCutWhardPWDelta} that \LCut is W[1]-hard for the combined parameter pathwidth and maximum degree of the input graph~$G$.
This is a nontrivial strengthening of the reduction provided by Dvořák and Knop~\cite{DvorakK18}, where the degree cannot be bounded by a function of the parameter.
Furthermore, our reduction implies that assuming ETH, there is no \( f(k) \cdot n^{o(k)} \)-time algorithm for \LCut, where $k$ is pathwidth of the input graph (whereas the reduction of Dvořák and Knop refutes only \(f(k) \cdot n^{\sqrt{k}}\)-time algorithms).
This implies that the algorithm of Kolman~\cite{Kolman18} with running-time \( n^{k} \) is optimal (such a bound can be obtained using the trivial bound $\lambda \le n$).
Moreover, this hardness result constitutes a natural counterpart of the known XP algorithm for the parameter maximum degree~\cite{BazganFNNS19}.

\subsection{Preliminaries}
For a given positive integer~$a$ we use~$[a]$ to denote the set~$\{1,2,\ldots,a\}$.
We mostly use standard graph notation.
We identify specific paths by just some of their vertices, e.g.\ we use the name \emph{$a$-$b$-$c$-path} to denote a path that starts in~$a$, then continues by some shortest~$a$-$b$-path and ends with some shortest~$b$-$c$ path. The shortest paths between two consecutive vertices in our identifiers ($a$-$b$ and~$b$-$c$ in our example) will always be unique.
We use~$G[X]$ to denote the \emph{induced subgraph} of a set~$X$ of vertices in a graph~$G$ and~$G-X$ to denote~$G[V\setminus X]$.
An \emph{interval graph} is a graph~$G = (V,E)$ such that each vertex~$v$ can be represented by an interval~$[\s{v},\e{v}]$ and two vertices~$u,w$ are adjacent in~$G$ if and only if~$[\s{u},\e{u}] \cap [\s{w},\e{w}] \neq \emptyset$.
A \emph{proper interval graph} is an interval graph such that there are no two vertices $v$ and $w$ such that $[\s{v}, \e{v}] \subseteq [\s{w}, \e{w}]$.
Equivalently, a proper interval graph can be defined as an interval graph where each interval has length one, i.e., $\s{v} + 1 = \e{v}$ for each vertex $v$ (see e.g.~\cite{BrandstadtLS99}).
An example of an interval graph and its interval representation is given in \Cref{fig:intervalgraph}.
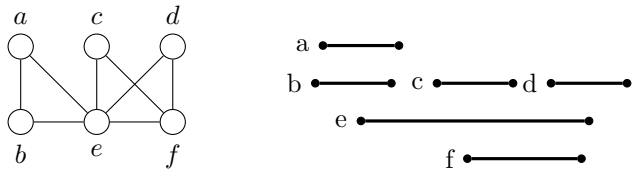
\begin{figure}
\begin{minipage}{.7\columnwidth}
	\centering
	\hfill
	\begin{tikzpicture}
	\node[circle,draw,label=$a$] at(0,1) (a) {};
	\node[circle,draw,label=below:$b$] at(0,0) (b) {} edge(a);
	\node[circle,draw,label=$c$] at(1,1) (c) {};
	\node[circle,draw,label=$d$] at(2,1) (d) {};
	\node[circle,draw,label=below:$e$] at(1,0) (e) {} edge(a) edge(b) edge(c) edge(d);
	\node[circle,draw,label=below:$f$] at(2,0) (f) {} edge(c) edge(d) edge(e);
	\end{tikzpicture}
	\hfill
	\begin{tikzpicture}
	\tikzDrawIntervall{1.5}{0}{1}{}{$a$}
	\tikzDrawIntervall{1}{-.1}{.9}{}{$b$}
	\tikzDrawIntervall{1}{1.5}{2.5}{}{$c$}
	\tikzDrawIntervall{1}{3}{4}{}{$d$}
	\tikzDrawIntervall{.5}{.5}{3.5}{}{$e$}
	\tikzDrawIntervall{0}{1.9}{3.4}{}{$f$}
	\end{tikzpicture}
	\hfill
	\hfill
\end{minipage}
\begin{minipage}{.29\columnwidth}
	\caption{An example of an interval graph (left side) and its interval representation (right side).}
	\label{fig:intervalgraph}
\end{minipage}
\end{figure}

A parameterization for a problem~$P$ is formally a pair of functions~$(f,g)$ such that~$f$ maps each possible input~$I$ for~$P$ to some object~$f(I)$ and~$g$ maps each such object to a non-negative integer.
One of the most prominent examples in the context of graph problems is the \emph{treewidth} of a graph.
Here,~$f$ maps each graph to a \emph{tree decomposition} of~$G$ and~$g$ measures the \emph{width} of the tree decomposition, that is, the maximum number of vertices in any bag of the tree decomposition (minus one).
A \emph{parameter} is then the resulting positive integer~$g(f(I))$ of a parameterization.
In this work we consider the parameters pathwidth, maximum degree and feedback vertex number.
The \emph{pathwidth} of a graph is closely related to the treewidth.
The only difference between the two concepts is that a path decomposition is restricted to a collection of paths as underlying graphs for the bags as opposed to forests for tree decompositions.
Formally, it is defined as follows.
\begin{definition}
  A \emph{path decomposition} of a graph $G=(V,E)$ consists of a graph $P=(W,F)$ which is a collection of disjoint paths and a function $\pi\colon W \rightarrow 2^V$ such that
  \begin{itemize}
    \item $\bigcup_{x\in W} \pi (x) = V$,
    \item for each edge $\{v, w\}\in E$, there exists an $x\in W$ such that $v,w\in \pi (x)$, and
    \item for each $v\in V$, we have that $\{x\in W \mid v \in \pi (x)\}$ induces a connected subgraph in $P$.
  \end{itemize}

  The \emph{width of a path decomposition} is defined as $\max_{x\in V} |\pi (x)| -1$.

  The \emph{pathwidth} of a graph $G$ is the minimum width of a path decomposition of $G$.
\end{definition}

The \emph{maximum degree} of a graph is the maximum number of incident edges to any single vertex in the graph.
The \emph{feedback vertex number} is the size of a minimum feedback vertex set, i.e., the minimum number of vertices one needs to delete from the graph to obtain a forest.

A problem~$P$ is \emph{fixed-parameter tractable} (or FPT for short) with respect to some \emph{parameter}~$k$ if there is an algorithm deciding~$P$ in~$f(k) \cdot \poly(n)$ time, where~$f$ is some computable function and~$n$ is the input size.
To show that some problem is \emph{presumably not} FPT with respect to some parameter, one regularly uses the standard complexity assumption that FPT $\neq$ W[1] and shows that a problem is \emph{W[1]-hard}.
To show W[1]-hardness for some problem~$P$ with respect to some parameter~$k$, one uses a parameterized reduction from some W[1]-hard problem~$Q$ with respect to some parameter~$\ell$.
A function~$R \colon \Sigma^*_Q \times \mathbb{N} \rightarrow \Sigma^*_P \times \mathbb{N}$ is a parameterized reduction if it for each instance~$(q,\ell)$ of~$Q$ produces an instance~$(p,k)$ of~$P$ such that
\begin{enumerate}
\item $(p,k)$ can be computed in~$f(\ell) \cdot \poly(|q|)$ time for some computable function~$f$,
\item $(p,k)$ is a yes-instance if and only if $(q,\ell)$ is a yes instance,
\item $k \leq g(\ell)$ for some computable function~$g$.
\end{enumerate}
The Exponential-Time Hypothesis (ETH) of Impagliazzo and Paturi~\cite{ImpagliazzoP01} asserts that there is no~\( 2^{o(m)} \) algorithm solving the \textsc{Satisfiability} problem, where $m$ is the number of clauses.
It is worth noting that assuming ETH, there is no \( f(k) \cdot n^{o(k)} \) time algorithm solving $k$-\textsc{(Multicolored) Clique}~\cite{ChenCFHJKX04}, where $f$ is a computable function and $k$ is the size of the clique we are looking for.
For further notions related to parameterized complexity and ETH, we refer the reader to the textbook by Cygan et al.~\cite{CyganFKLMPPS15}.

\section{W[1]-Hardness for Pathwidth and Maximum Degree}
\label{sec:PWMD}
In this section we will prove that \lbc{} is W[1]-hard with respect to the combined parameter pathwidth and maximum degree.
We will start by describing our reduction from the W[1]-hard \textsc{Clique} problem parameterized by solution size and then prove that it is correct and fulfills running-time and parameter-size constraints.
\defProblemTask{\textsc{Clique}}
{An undirected graph $G$ and an integer $k$.}
{Decide whether $G$ contains a clique of size at least $k$.}

To formulate our reduction we will use~$(G=(V,E), k)$ as an input instance of \textsc{Clique} parameterized by solution size.
We assume that we are also given a bijection $\idx : V \rightarrow [n]$, and call its inverse $\vertex : [n] \rightarrow V$.
We order the edges $E = \{e_1, e_2, \dots, e_m\}$ lexicographically by their endpoints, that is, for~$e_i = \{v_i, w_i\}$ and~$e_{i+1} = \{v_{i+1}, w_{i+1}\}$, we have~$\idx (v_j) < \idx (w_j)$ for~$j\in \{i, i+1\}$ and either~$\idx (v_i) < \idx (v_{i+1})$ or ($\idx (v_i) = \idx (v_{i+1})$ and $\idx (w_i) < \idx (w_{i+1})$).

Furthermore, we assume without loss of generality that $m\ge n$ as we can otherwise remove all connected components which are trees (possibly returning true if $k\le 2$), as the largest clique in such a component is of size at most two (which is tight if and only if the component contains an edge).

Let $\eta = 4m$.
We construct a \lbc{} instance $(H, s, t, \beta, \lambda)$ as follows.
We set the budget $\beta$ of edges to delete to $ 2k + 2 k(k-1)= 2k^2$ and the length $\lambda:= 8 \eta + 2n + 1$.
The graph~$H$ will consist of vertex-selection gadgets, incidence-checking gadgets, connectivity paths, and the vertices $s$ and $t$, which are not contained in any gadget.
This will then allow us to prove the following theorem.

\begin{restatable}{theorem}{pwmdW}
  \label{thm:LCutWhardPWDelta}
  \textsc{Length-Bounded Cut} parameterized by the combined parameter pathwidth and maximum degree $k$ is W[1]-hard.
  Assuming ETH, it cannot be solved in ${f(\sqrt{k}+\pw(G))}\cdot {n^{o(\sqrt{k}+\pw(G))}}$ time for any computable function~$f$, where $\pw (G) $ is the pathwidth of~$G$.
\end{restatable}

\subparagraph*{Vertex-Selection Gadgets.}
Our reduction will produce $k$ vertex-selection gadgets $A_1, \dots, A_k$.
The gadget $A_j$ looks as follows:

Start with two paths $u^j_0,\dots, u^j_n$ and $\ell^j_0, \dots, \ell^j_n$ of length $n$ (``upper'' and ``lower'' path).
Add an~$u^j_{p-1}$-$u^j_{p}$-path $U^j_p$ of length $2\eta + p$, a $\ell^j_{p-1}$-$\ell^j_{p}$-path $L^j_p$ of length $2\eta-p$, and a $u^j_p$-$\ell^j_p$-path of length~$2 \eta$ for every $p\in [n]$.
Finally, add two paths of length $\eta+2 $ between the source $s$ and the first vertex~$\ell^j_0$ of the lower path and two paths of length two between $s$ and the first vertex~$u^j_0$ of the upper path (see Figure~\ref{fvertexselection}).

We call the vertex $u^j_n$ the \emph{end vertex of the upper path} and $\ell^j_n$ the \emph{end vertex of the lower path}.

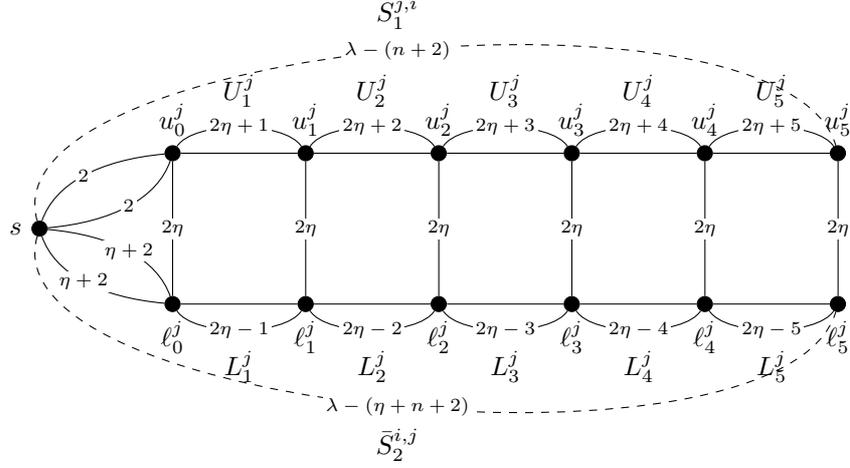
\begin{figure}[t]
  \begin{center}
    \begin{tikzpicture}[xscale=1.75]
      \node[vertex, label=180:$s$] (s) at (0, 0) {};

      \node[vertex, label=90:$u^j_0$] (v1) at (1, 1) {};
      \node[vertex, label=90:$u^j_1$] (v2) at (2, 1) {};
      \node[vertex, label=90:$u^j_2$] (v3) at (3, 1) {};
      \node[vertex, label=90:$u^j_3$] (v4) at (4, 1) {};
      \node[vertex, label=90:$u^j_4$] (v5) at (5, 1) {};
      \node[vertex, label=90:$u^j_5$] (v6) at (6, 1) {};

      \draw (s) edge[bend left = 35] node[pos=0.5, fill=white, inner sep=1pt] {\scriptsize $2$} (v1);
      \draw (s) edge[bend right = 35] node[pos=0.5, fill=white, inner sep=1pt] {\scriptsize $2$} (v1);
      \draw (v1)--(v2)--(v3)--(v4)--(v5)--(v6);

      \draw (v1) edge[bend left = 70] node[pos=0.5, fill=white, inner sep=1pt, label=90:$U^j_1$] {\scriptsize $2\eta+1$} (v2);
      \draw (v2) edge[bend left = 70] node[pos=0.5, fill=white, inner sep=1pt, label=90:$U^j_2$] {\scriptsize $2\eta+2$} (v3);
      \draw (v3) edge[bend left = 70] node[pos=0.5, fill=white, inner sep=1pt, label=90:$U^j_3$] {\scriptsize $2\eta+3$} (v4);
      \draw (v4) edge[bend left = 70] node[pos=0.5, fill=white, inner sep=1pt, label=90:$U^j_4$] {\scriptsize $2\eta+4$} (v5);
      \draw (v5) edge[bend left = 70] node[pos=0.5, fill=white, inner sep=1pt, label=90:$U^j_5$] {\scriptsize $2\eta+5$} (v6);

      \node[vertex, label=270:$\ell^j_0$] (w1) at (1, -1) {};
      \node[vertex, label=270:$\ell^j_1$] (w2) at (2, -1) {};
      \node[vertex, label=270:$\ell^j_2$] (w3) at (3, -1) {};
      \node[vertex, label=270:$\ell^j_3$] (w4) at (4, -1) {};
      \node[vertex, label=270:$\ell^j_4$] (w5) at (5, -1) {};
      \node[vertex, label=270:$\ell^j_5$] (w6) at (6, -1) {};

      \draw (s) edge[bend left = 35] node[pos=0.5, fill=white, inner sep=1pt] {\scriptsize $\eta+ 2$} (w1);
      \draw (s) edge[bend right = 35] node[pos=0.5, fill=white, inner sep=1pt] {\scriptsize $\eta+ 2$} (w1);
      \draw (w1) --(w2)--(w3)--(w4)--(w5)--(w6);

      \draw (w1) edge[bend right = 70] node[pos=0.5, fill=white, inner sep=1pt, label=270:$L^j_1$] {\scriptsize $2\eta-1$} (w2);
      \draw (w2) edge[bend right = 70] node[pos=0.5, fill=white, inner sep=1pt, label=270:$L^j_2$] {\scriptsize $2\eta-2$} (w3);
      \draw (w3) edge[bend right = 70] node[pos=0.5, fill=white, inner sep=1pt, label=270:$L^j_3$] {\scriptsize $2\eta-3$} (w4);
      \draw (w4) edge[bend right = 70] node[pos=0.5, fill=white, inner sep=1pt, label=270:$L^j_4$] {\scriptsize $2\eta-4$} (w5);
      \draw (w5) edge[bend right = 70] node[pos=0.5, fill=white, inner sep=1pt, label=270:$L^j_5$] {\scriptsize $2\eta-5$} (w6);

      \draw (v1) edge node[pos=0.5, fill=white, inner sep=1pt] {\scriptsize $2\eta$} (w1);
      \draw (v2) edge node[pos=0.5, fill=white, inner sep=1pt] {\scriptsize $2 \eta$} (w2);
      \draw (v3) edge node[pos=0.5, fill=white, inner sep=1pt] {\scriptsize $2 \eta$} (w3);
      \draw (v4) edge node[pos=0.5, fill=white, inner sep=1pt] {\scriptsize $2 \eta$} (w4);
      \draw (v5) edge node[pos=0.5, fill=white, inner sep=1pt] {\scriptsize $2 \eta$} (w5);
      \draw (v6) edge node[pos=0.5, fill=white, inner sep=1pt] {\scriptsize $2 \eta$} (w6);

      \draw (s) edge[bend left = 90, dashed] node[pos=0.5, fill=white, inner sep=1pt, label=90:$S^{j,i}_1$] {\scriptsize $\lambda -(n+2)$} (v6);
      \draw (s) edge[bend right = 90, dashed] node[pos=0.5, fill=white, inner sep=1pt, label=270:$\bar{S}^{i,j}_2$] {\scriptsize $\lambda -(\eta + n+2)$} (w6);
    \end{tikzpicture}
  \end{center}
  \caption{An example of a vertex-selection gadget with $n = 5$.
  An edge with a number $x$ on it corresponds to a path of length $x$.
  The connectivity paths $S^{j, i}_1$ and $\bar{S}^{j, i}_2$ are dashed (all other connectivity paths are not drawn for the simplicity of the picture).}
  \label{fvertexselection}
\end{figure}

\subparagraph*{Incidence-Checking Gadgets.}
For each pair of vertex-selection gadgets $(A_i, A_j)$ with~${i< j}$, we add an incidence-checking gadget $I^{i,j}$.
Starting at $u^i_n$ and $\ell^i_n$, we add a gadget similar to the vertex-selection gadget.
More precisely, we have two paths $a^{i,j}_0, \dots, a^{i,j}_n$ and $b^{i,j}_0, \dots, b^{i,j}_n$.
The vertex~$u^i_n$ is connected to~$a^{i,j}_0$ by two paths of length $4\eta$.
The union of the edges of these two paths is denoted by $E^{i, j}_a$.
The vertex~$\ell^i_n$ is connected to $b^{i,j}_0$ by two paths of length two.
The edges inside these two paths are denoted by $E^{i,j}_{b}$.
For each $p\in \{0, 1, \ldots, n\}$, there is an~$a^{i,j}_p$-$b^{i,j}_p$-path of length~$4\eta$, an~$a^{i,j}_{p-1}$-$a^{i,j}_{p}$-path~$A^{i, j}_p$ of length~$\eta-p$ and a $b^{i,j}_{p -1}$-$b^{i,j}_{p}$-path $B^{i, j}_p$ of length $\eta + p$.
Furthermore, there are two $a^{i,j}_n$-$t$-paths of length two, and two $b^{i,j}_n$-$t$-paths of length $3\eta$.

Next we add two paths $c^{i,j}_0, \dots, c^{i,j}_m$ and $d^{i,j}_0, \dots, d^{i,j}_m$ of length $m$.
The vertex $c^{i,j}_0 $ is connected to~$u^j_n$ by two parallel paths of length~$3\eta$; the union of their edges is denoted by~$E^{i,j}_c$.
The vertex~$d^{i,j}_0$ is connected to~$\ell^j_n$ by two paths of length $\eta$; the union of their edges is denoted by~$E^{i,j}_d$.
For each~$p\in [m]$, let $e_p = \{v_p, w_p\}$ with~$\idx (v_p) < \idx (w_p)$, there is an $c^{i,j}_p$-$d^{i,j}_p$-path of length $2 \eta$, and an $c^{i,j}_{p-1}$-$c^{i,j}_p$-path $C^{i,j}_p$ of length $2\eta + \idx (w_p)$, and a $d^{i,j}_{p-1}$-$d^{i,j}_p$-path~$D^{i,j}_p$ of length~$2\eta - w_p$.
Furthermore, there are two $c^{i,j}_n$-$t$-paths of length~$\eta + n -m + 2$, and two~$b^{i,j}_n$-$t$-paths of length~$2\eta + n - m + 2$.

For each $p \in \{0, 1, \ldots, n-1\}$, let $q$ be zero if $p=0$, and otherwise maximum such that~$e_q$ is incident to $\vertex (p)$.
We add an~$a^{i,j}_p$-$c^{i,j}_q$-path of length $3 \eta$, an $a^{i,j}_p$-$d^{i,j}_q$-path of length~$2\eta$, a~$b^{i,j}_p$-$c^{i,j}_q$-path of length~$2 \eta$ and a~$b^{i,j}_p$-$d^{i,j}_q$-path of length~$3 \eta$.
Note that this construction is not symmetric in $i$ and $j$.
An example of an incidence-checking gadget can be found in Figure~\ref{fincidence}.

\begin{figure}[t]
  \begin{center}
	\resizebox{\textwidth}{!}{
	    \begin{tikzpicture}[xscale=2.5, yscale = 1.75]
	      \definecolor{darkred}{rgb}{.6,0,0}
	      \node[vertex, label=0{\Large :$t$}] (s) at (7, -2) {};

	      \node[vertex, label=180:{\Large $u^i_n$}] (vn) at (0, 1) {};
	      \node[vertex, label=180:{\Large $\ell^i_n$}] (wn) at (0, -1) {};

	      \node[vertex, label=90:{\Large $a^{i,j}_0$}] (v1) at (1, 1) {};
	      \node[vertex, label=90:{\Large $a^{i,j}_1$}] (v2) at (3, 1) {};
	      \node[vertex] (v3) at (3, 1) {};
	      \node[vertex] (v4) at (5, 1) {};
	      \node[vertex, label=90:{\Large $a^{i,j}_2$}] (v5) at (5, 1) {};
	      \node[vertex, label=90:{\Large $a^{i,j}_3$}] (v6) at (6, 1) {};

	      \draw (s) edge[bend right = 45] node[pos=0.5, fill=white, inner sep=1pt] {$2$} (v6);
	      \draw (s) edge[bend right = 25] node[pos=0.5, fill=white, inner sep=1pt] {$2$} (v6);
	      \draw (v1)--(v2)--(v3)--(v4)--(v5)--(v6);

	      \draw (vn) edge[bend left = 70] node[pos=0.5, fill=white, inner sep=1pt] {$4\eta$} (v1);
	      \draw (vn) edge[bend right = 70] node[pos=0.5, fill=white, inner sep=1pt] {$4\eta$} (v1);

	      \draw (v1) edge[bend left = 70] node[pos=0.5, fill=white, inner sep=1pt] {$2\eta-1$} (v2);
	      \draw (v3) edge[bend left = 70] node[pos=0.5, fill=white, inner sep=1pt] {$2\eta-2$} (v4);
	      \draw (v5) edge[bend left = 70] node[pos=0.5, fill=white, inner sep=1pt] {$2\eta-3$} (v6);

	      \node[vertex, label=270:{\Large $b^{i,j}_0$}] (w1) at (1, -1) {};
	      \node[vertex, label=270:{\Large $b^{i,j}_1$}] (w2) at (3, -1) {};
	      \node[vertex] (w3) at (3, -1) {};
	      \node[vertex] (w4) at (5, -1) {};
	      \node[vertex, label=270:{\Large $b^{i,j}_2$}] (w5) at (5, -1) {};
	      \node[vertex, label=270:{\Large $b^{i,j}_3$}] (w6) at (6, -1) {};

	      \draw (wn) edge[bend left = 70] node[pos=0.5, fill=white, inner sep=1pt] {$2$} (w1);
	      \draw (wn) edge[bend right = 70] node[pos=0.5, fill=white, inner sep=1pt] {$2$} (w1);

	      \draw (s) edge[bend left = 35] node[pos=0.5, fill=white, inner sep=1pt] {$3\eta$} (w6);
	      \draw (s) edge[bend right = 35] node[pos=0.5, fill=white, inner sep=1pt] {$3\eta$} (w6);
	      \draw (w1) edge (w3);
	      \draw (w3) edge (w5);
	      \draw (w5) edge (w6);

	      \draw (w1) edge[bend right = 70] node[pos=0.5, fill=white, inner sep=1pt] {$2\eta+1$} (w2);
	      \draw (w3) edge[bend right = 70] node[pos=0.5, fill=white, inner sep=1pt] {$2\eta+2$} (w4);
	      \draw (w5) edge[bend right = 70] node[pos=0.5, fill=white, inner sep=1pt] {$2\eta+3$} (w6);

	      \draw (v1) edge node[pos=0.5, fill=white, inner sep=1pt] {$4\eta$} (w1);
	      \draw (v2) edge node[pos=0.5, fill=white, inner sep=1pt] {$4\eta$} (w2);
	      \draw (v3) edge node[pos=0.5, fill=white, inner sep=1pt] {$4\eta$} (w3);
	      \draw (v4) edge node[pos=0.5, fill=white, inner sep=1pt] {$4\eta$} (w4);
	      \draw (v5) edge node[pos=0.5, fill=white, inner sep=1pt] {$4\eta$} (w5);
	      \draw (v6) edge node[pos=0.5, fill=white, inner sep=1pt] {$4\eta$} (w6);

	    \draw (vn) edge[dashed] node[pos=0.28, fill=white, inner sep=1pt] {$\lambda - (4\eta + n + 2)$} (s);

	    \draw (wn) edge[dashed] node[pos=0.28, fill=white, inner sep=1pt] {$\lambda - (3\eta + n + 2)$} (s);
	    \begin{scope}[yshift = -4cm]

	      \node[vertex, label=90:{\Large $c^{i,j}_0$}] (v1) at (1, 1) {};
	      \node[vertex, label=90:{\Large $c^{i,j}_1$}] (v2) at (2, 1) {};
	      \node[vertex, label=90:{\Large $c^{i,j}_2$}] (v3) at (3, 1) {};
	      \node[vertex, label=90:{\Large $c^{i,j}_3$}] (v4) at (4, 1) {};
	      \node[vertex, label=90:{\Large $c^{i,j}_4$}] (v5) at (5, 1) {};
	      \node[vertex, label=90:{\Large $c^{i,j}_5$}] (v6) at (6, 1) {};

	      \draw (s) edge[bend left = 35] node[pos=0.5, fill=white, inner sep=1pt] {$\eta + n -m+ 2$} (v6);
	      \draw (s) edge[bend right = 35] node[pos=0.5, fill=white, inner sep=1pt] {$\eta + n - m + 2$} (v6);
	      \draw (v1)--(v2)--(v3)--(v4)--(v5)--(v6);

	      \draw (v1) edge[bend left = 70] node[pos=0.5, fill=white, inner sep=1pt] {$2\eta-u_1$} (v2);
	      \draw (v2) edge[bend left = 70] node[pos=0.5, fill=white, inner sep=1pt] {$2\eta-u_2$} (v3);
	      \draw (v3) edge[bend left = 70] node[pos=0.5, fill=white, inner sep=1pt] {$2\eta-u_3$} (v4);
	      \draw (v4) edge[bend left = 70] node[pos=0.5, fill=white, inner sep=1pt] {$2\eta-u_4$} (v5);
	      \draw (v5) edge[bend left = 70] node[pos=0.5, fill=white, inner sep=1pt] {$2\eta-u_5$} (v6);

	      \node[vertex, label=270:{\Large $d^{i,j}_0$}] (w1) at (1, -1) {};
	      \node[vertex, label=270:{\Large $d^{i,j}_1$}] (w2) at (2, -1) {};
	      \node[vertex, label=270:{\Large $d^{i,j}_2$}] (w3) at (3, -1) {};
	      \node[vertex, label=270:{\Large $d^{i,j}_3$}] (w4) at (4, -1) {};
	      \node[vertex, label=270:{\Large $d^{i,j}_4$}] (w5) at (5, -1) {};
	      \node[vertex, label=270:{\Large $d^{i,j}_5$}] (w6) at (6, -1) {};

	      \node[vertex, label=180:{\Large $u^j_n$}] (vn) at (0, 1) {};
	      \node[vertex, label=180:{\Large $\ell^j_n$}] (wn) at (0, -1) {};

	      \draw (vn) edge[bend left = 70] node[pos=0.5, fill=white, inner sep=1pt] {$3\eta$} (v1);
	      \draw (vn) edge[bend right = 70] node[pos=0.5, fill=white, inner sep=1pt] {$3\eta$} (v1);

	      \draw (wn) edge[bend left = 70] node[pos=0.5, fill=white, inner sep=1pt] {$\eta$} (w1);
	      \draw (wn) edge[bend right = 70] node[pos=0.5, fill=white, inner sep=1pt] {$\eta$} (w1);

	      \draw (s) edge[bend left = 5] node[pos=0.5, fill=white, inner sep=1pt] {$2\eta + n - m + 2$} (w6);
	      \draw (s) edge[bend left = 55] node[pos=0.5, fill=white, inner sep=1pt] {$2\eta + n - m + 2$} (w6);
	      \draw (w1) --(w2)--(w3)--(w4)--(w5)--(w6);

	      \draw (w1) edge[bend right = 70] node[pos=0.5, fill=white, inner sep=1pt] {$2\eta+u_1$} (w2);
	      \draw (w2) edge[bend right = 70] node[pos=0.5, fill=white, inner sep=1pt] {$2\eta+u_2$} (w3);
	      \draw (w3) edge[bend right = 70] node[pos=0.5, fill=white, inner sep=1pt] {$2\eta+u_3$} (w4);
	      \draw (w4) edge[bend right = 70] node[pos=0.5, fill=white, inner sep=1pt] {$2\eta+u_4$} (w5);
	      \draw (w5) edge[bend right = 70] node[pos=0.5, fill=white, inner sep=1pt] {$2\eta+u_5$} (w6);

	      \draw (v1) edge node[pos=0.5, fill=white, inner sep=1pt] {$2\eta$} (w1);
	      \draw (v2) edge node[pos=0.5, fill=white, inner sep=1pt] {$2\eta$} (w2);
	      \draw (v3) edge node[pos=0.5, fill=white, inner sep=1pt] {$2\eta$} (w3);
	      \draw (v4) edge node[pos=0.5, fill=white, inner sep=1pt] {$2\eta$} (w4);
	      \draw (v5) edge node[pos=0.5, fill=white, inner sep=1pt] {$2\eta$} (w5);
	      \draw (v6) edge node[pos=0.5, fill=white, inner sep=1pt] {$2\eta$} (w6);
	    \end{scope}

	    \draw (1, 1) edge[very thick,dotted,green, bend left] node[pos=0.45, fill=white, inner sep=1pt] {\color{green} $2\eta$} ($(1, -3)$);
	    \draw (3, 1) edge[very thick,dotted,green, bend left] node[pos=0.45, fill=white, inner sep=1pt] {\color{green} $2\eta$} ($(3, -3)$);
	    \draw (5, 1) edge[very thick,dotted,green, bend left] node[pos=0.45, fill=white, inner sep=1pt] {\color{green} $2\eta$} ($(5, -3)$);

	    \draw (1, 1) edge[very thick,dotted,green, bend right] node[pos=0.5, fill=white, inner sep=1pt] {\color{green} $3\eta$} ($(1, -5)$);
	    \draw (3, 1) edge[very thick,dotted,green, bend right] node[pos=0.5, fill=white, inner sep=1pt] {\color{green} $3\eta$} ($(3, -5)$);
	    \draw (5, 1) edge[very thick,dotted,green, bend right] node[pos=0.5, fill=white, inner sep=1pt] {\color{green} $3\eta$} ($(5, -5)$);

	    \begin{scope}[yshift = -2 cm]

	      \draw (1, 1) edge[very thick,dotted,darkred, bend left] node[pos=0.5, fill=white, inner sep=1pt] {\color{darkred} $3\eta$} ($(1, -1)$);
	      \draw (3, 1) edge[very thick,dotted,darkred, bend left] node[pos=0.5, fill=white, inner sep=1pt] {\color{darkred} $3\eta$} ($(3, -1)$);
	      \draw (5, 1) edge[very thick,dotted,darkred, bend left] node[pos=0.5, fill=white, inner sep=1pt] {\color{darkred} $3\eta$} ($(5, -1)$);

	    \draw (1, 1) edge[very thick,dotted,darkred, bend right] node[pos=0.25, fill=white, inner sep=1pt] {\color{darkred} $2\eta$} ($(1, -3)$);
	    \draw (3, 1) edge[very thick,dotted,darkred, bend right] node[pos=0.25, fill=white, inner sep=1pt] {\color{darkred} $2\eta$} ($(3, -3)$);
	    \draw (5, 1) edge[very thick,dotted,darkred, bend right] node[pos=0.25, fill=white, inner sep=1pt] {\color{darkred} $2\eta$} ($(5, -3)$);
	    \end{scope}

	    \draw (vn) edge[dashed] node[pos=0.58, fill=white, inner sep=1pt] {$\lambda - (4\eta + n + 2)$} (s);

	    \draw (wn) edge[dashed] node[pos=0.49, fill=white, inner sep=1pt] {$\lambda - (3\eta + n + 2)$} (s);

	    \end{tikzpicture}
	}
  \end{center}
  \caption{An example of an incidence-checking gadget $I^{i,j}$.
  An edge with a number $x$ on it corresponds to a path of length $x$.
  Paths between~$a,b$-vertices and~$c,d$-vertices are dotted and colored to be better distinguishable from other paths.
  Connectivity paths $T^{i}_1$, $T^{j}_1$, $\bar{T}^{i}_1$, and $\bar{T}^{j}_1$ are dashed and connectivity paths $T^{i}_2$, $T^{j}_2$, $\bar{T}^{i}_2$, and $\bar{T}^{j}_2$ are not shown for readability.
  }
  \label{fincidence}
\end{figure}
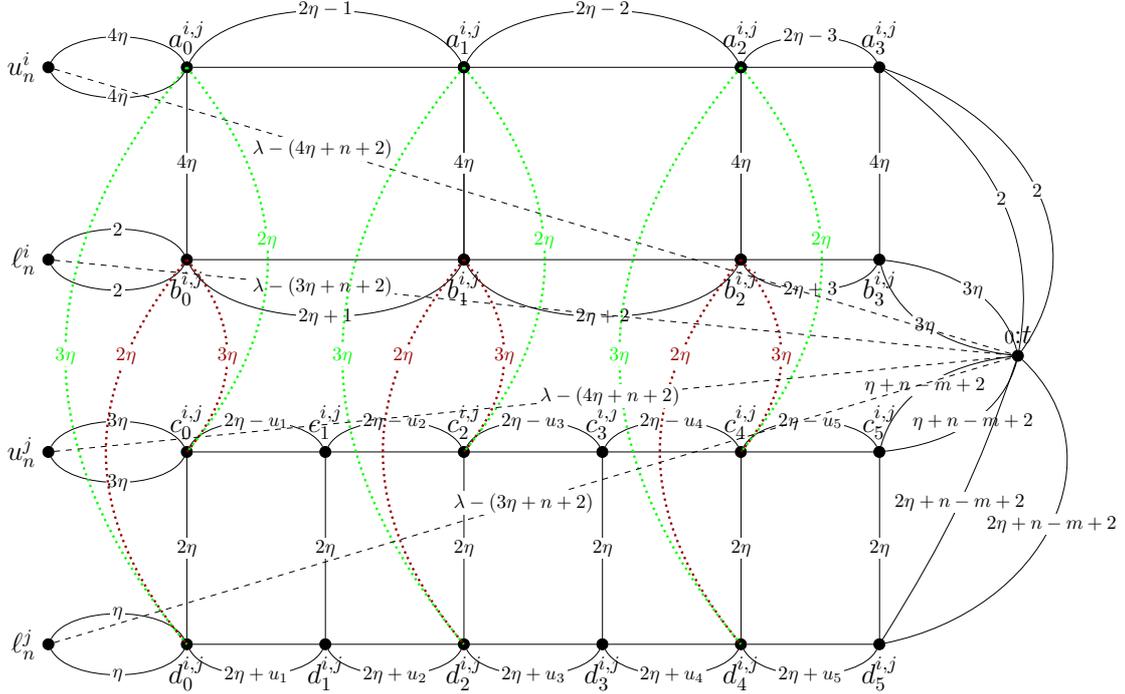

\subparagraph*{Connectivity Paths.}
For each end vertex~$u^i_n$ of an upper path of a vertex-selection gadget~$A_i$, we add three parallel paths $T^i_1$, $T^i_2$, and~$T^i_3$ of length $\lambda - (n + 2)$ to $t$.
For each end vertex~$\ell^i_n$ of a lower path of a vertex-selection gadget~$A_i$, we add three parallel paths $\bar{T}^i_1$, $\bar{T}^i_2$, and $\bar{T}^i_3$ of length~$\lambda - (\eta + 2 + n)$ to $t$ (see Figure~\ref{fincidence}).
For each $(i, j)\in [n]^2$, we add five parallel paths~$S^{i, j}_q$ from $s$ to $u^i_n$ of length $\lambda - (4\eta + n +2)$, and five parallel paths~$\bar{S}^{i, j}_q$ from $s$ to $\ell^i_n$ of length $\lambda - (3\eta + n +2)$ (see Figuer~\ref{fvertexselection}).
We call the resulting graph $H$.

Before we present the proofs, we start with some basic notation.
\begin{notation}
  We denote by $x_0$-$X_p$-$x_q$ the path consisting of $x_0,x_1, \dots, x_{p-1}$, then the $x_{p-1}$-$x_p$-path of length $2\eta \pm x$, and finally $x_{p+1}, x_{p+2}, \dots, x_q$ for $x\in \{u^i, \ell^i, a^{i,j}, b^{i,j}, c^{i,j}, d^{i,j}\}$.

	  Similarly, $x_p$-$y_p$ will denote the~$x_p$-$y_p$-path of length~$2 \eta$ for~$x= u^i$ and~$y = \ell^i$ or~$x=d^{i,j}$ and~$y = c^{i,j}$.
  For $x=b^{i,j}$ and $y=a^{i,j}$, it denotes teh $x_p$-$y_p$-path of length $4\eta$.

	  For $p\in [n]$ and $q\in [m]$ such that there exists a $b^{i,j}_p$-$c^{i,j}_q$-path of length $3\eta$, we denote this path by $b^{i,j}_p$-$c^{i,j}_q$.
  We denote $b^{i,j}_p$-$d^{i,j}_q$-paths of length $2\eta$, $c^{i,j}_q$-$a^{i,j}_p$-paths of length $2\eta$, and~$d^{i,j}_q$-$a^{i,j}_p$-paths of length $3\eta$ in a similar way.

		We denote the set of edges contained in a connectivity path $T^j_\ell$, $\bar{T}^j_\ell$, $S^{i,j}_\ell$, $\bar{S}^{i,j}_\ell$, a vertex-selection gadget $A_i$ or an incidence-checking gadget $I^{i,j}$ by $E(T^j_\ell)$, $E(\bar{T}^j_\ell)$, $E(S^{i,j}_\ell)$, $E(\bar{S}^{i,j}_\ell)$, $E(A_i)$, or~$E(I^{i,j})$, respectively.

		Finally, we say that we \emph{suppress} a degree-two vertex $v$, if we contract it with one of its neighbors.
\end{notation}

We now prove the correctness of our reduction and afterwards analyze the running time needed to compute it and finally investigate the size of the pathwidth and the maximum degree in the resulting graph.
\subsection{Forward Direction}

		We will show that any $\lambda$-cut in $H$ of size at most $\beta$ implies a clique of size $k$ in $G$.

		First, we show that $F$ has to contain at least two edges in any ``front part'' of a vertex selection gadget (or connectivity paths).

		\begin{lemma}
  \label{lTwoEdgesInVSG}
  Let $F$ be any $\lambda$-cut.
  Then for each $i\in [n]$, the cut $F$ contains at least two edges from $E(A_j) \cup E(T^j_1) \cup E(T^j_2)$.
  If $F$ contains exactly two edges from $E(A_j) \cup E(T^j_1)\cup E(T^j_2)$, then $F$ contains one edge of the form $\{u_{p-1}^j, u_{p}^j\}$ and one edge of the form $\{\ell_{q-1}^j, \ell_q^j\}$.
\end{lemma}

		\begin{proof}
  If $F$ contains at least three edges from $E(A_j) \cup E(T^j_1)\cup E(T^j_2)\cup E(T^j_3)$, then there is nothing to show, so we assume that $F$ contains at most two edges from $E(A_j) \cup E(T^j_1)\cup E(T^j_2)\cup E(T^j_3)$.
  We need to show that one of the two edges is of the form $\{u^j_{i-1}, u^j_{i}\}$, while the other edge is of the form $\{\ell^j_{i-1}, \ell^j_{i}\}$.

\noindent{\bfseries Case 1 ($F$} contains no edge {\bfseries $\{u^j_{i-1}, u^j_{i}\}$):}
		  If~$F$ does not contain an edge from both $s$-$u^j_0$-paths of length two, then there is a path~$s$-$u^j_0$-$u^j_n$ of length~$n + 2$ from~$s$ to~$u^j_n$.
  This together with the connectivity path~$T^j_1$,~$T^j_2$, or~$T^j_3$ from~$u^j_n$ to~$t$ yields an~$s$-$u^j_0$-$u^j_n$-$T^j_{1/2}$-$t$-path of length~$n  + 2 + {\lambda - (n + 2)} = \allowbreak \lambda$, so one edge from both~$T^j_1$ and~$T^j_2$ or one edge from both~$s$-$ u^j_0$-paths need to be contained in~$F$.

		  Thus, $F$ contains an edge of the form $\{u^j_{i-1}, u^j_{i}\}$.

\noindent{\bfseries Case 2 ($F$} contains no edge {\bfseries $\{\ell^j_{i-1}, \ell^j_{i}\}$):}
		  If $F$ does not contain one edge from each of the two~$s $-$\ell^j_0$-paths of length $\eta + 2$, then there is a path $s$-$\ell_0^j$-$\ell_n^j$ of length~$\eta + 2 + n$ from $s$ to $\ell^j_n$.
  This together with the connectivity path~$\bar{T}_1^j$,~$\bar{T}^j_2$, or~$\bar{T}^j_3$ yields an~$s$-$t$-path~$s$-$\ell^j_0$-$\ell^j_n$-$\bar{T}^j_{1/2}$ of length $\eta + 2 + n + \lambda - (\eta + 2 + n) = \lambda$, so again one edge from both~$\bar{T}_1^j$ and~$\bar{T}^j_2$ or one edge from each of the two~$s$-$\ell^j_n$-paths of length~$\eta + 2$ need to be contained in~$F$.

		  Thus, $F$ contains an edge of the form $\{\ell^j_{i-1}, \ell^j_{i}\}$.
\end{proof}

		Next, we show that each~$\lambda$-cut contains at least four edges from each incidence-checking gadget (or their respective connectivity paths).

		\begin{lemma}
  \label{lFourEdgesInICG}
  Let $F$ be any $\lambda$-cut, and let $(i,j)\in [n]^2$ with $i < j$.
  Then $F$ contains at least four edges from the incidence-checking gadget $I^{i,j }$ and the connectivity paths $S^{i, j}_p$, $S^{j, i}_p$, $\bar{S}^{i,j }_p$, and $\bar{S}^{j, i}_p$ for~${p\in [5]}$.
  If $F$ contains exactly four edges from $$X:= E(I^{i, j}) \bigcup \cup_{p=1}^5 \bigl( E(S^{i, j}_p)\cup E(S^{j, i}_p)\cup E(\bar{S}^{i,j }_p) \cup E(\bar{S}^{j, i}_p)\bigr),$$ then $F$ contains one edge of each of the following four forms: $\{a^{i,j}_{p-1}, a^{i,j}_{p}\}$,  $\{b^{i,j}_{p-1}, b^{i,j}_{p}\}$, $\{c^{i,j}_{p-1}, c^{i,j}_{p}\}$, and $\{d^{i,j}_{p-1}, d^{i,j}_{p}\}$.
\end{lemma}

		\begin{proof}
  If $F$ contains at least five edges from $X$, then there is nothing to show.

		  Thus, we assume that $F$ contains at most four edges from $X$, and will show by case distinction that $F$ contains one edge of each of the four forms.

\noindent{\bfseries Case 1 ($F$} contains no edge $\{a^{i,j}_{p-1}, a^{i,j}_{p}\}$ (or $\{c^{i,j}_{p-1}, c^{i,j}_{p}\}$){\bfseries ):}
		  Then there is a path~\mbox{$u^i_n$-$a^{i,j}_0$-$a^{i,j}_n$-$t$} (or~\mbox{$u^j_n$-$c^{i,j}_0$-$c^{i,j}_n$-$t$}) of length~$4\eta + n + 2$ from~$u^i_n$ ($u^j_n$) to~$t$.
  This together with a connectivity path~$S^{i, j}_q$ ($S^{j, i}_q$) forms an $s$-$t$-path $s$-$S^{i,j}_{q}$-$u^i_n$-$a^{i,j}_0$-$a^{i,j}_n$-$t$ ($s$-$S^{i,j}_{q}$-$u^j_n$-$c^{i,j}_0$-$c^{i,j}_n$-$t$) of length~$\lambda - {(4\eta + n + 2)} + {4\eta + n + 2} = \lambda$.

		  Thus, $F$ contains an edge of the form $\{a^{i,j}_{p-1}, a^{i,j}_{p}\}$ ($\{c^{i,j}_{p-1}, c^{i,j}_{p}\}$).

\noindent{\bfseries Case 2 ($F$} contains no edge $\{b^{i,j}_{p-1},b^{i,j}_{p}\}$ (or $\{d^{i,j}_{p-1}, d^{i,j}_{p}\}$){\bfseries ):}
		  Then there is a path~\mbox{$\ell^i_n$-$b^{i,j}_0$-$b^{i,j}_m$-$t$} (or~\mbox{$\ell^j_n$-$d^{i,j}_0$-$d^{i,j}_m$-$t$}) of length $3\eta + n + 2$ from $\ell^i_n $ ($\ell^j_n$) to~$t$.
  This together with a connectivity path~$\bar{S}^{i, j}_q$ ($\bar{S}^{j, i}_q$) forms a~$s$-$t$-path $s$-$\bar{S}^{i,j}_{q}$-$\ell^i_n$-$b^{i,j}_0$-$b^{i,j}_m$-$t$ ($s$-$\bar{S}^{j,i}_{q}$-$\ell^i_n$-$d^{i,j}_0$-$d^{i,j}_m$-$t$) of length~$\lambda - {(3\eta + n + 2)} + {3\eta + n + 2} = \lambda$.

		  Thus, $F$ contains an edge of the form $\{b^{i,j}_{p-1},b^{i,j}_{p}\}$ ($\{d^{i,j}_{p-1}, d^{i,j}_{p}\}$).
\end{proof}

		Combining Lemmata~\ref{lTwoEdgesInVSG} and~\ref{lFourEdgesInICG}, we get structural properties of any $\lambda$-cut $F$ of size at most $\beta$.
We show that any minimal $\lambda$-cut of size at most $\beta$ does not contain edges in connectivity paths.

		\begin{lemma}\label{loneedgeperpath}
  Let $F$ be any $\lambda$-cut with $|F|\le \beta$.
  For every vertex-selection gadget, $F$ contains exactly one edge in its upper and exactly one in its lower path.
  For every incidence-checking gadget $I^{i,j}$, it contains exactly one edge of each of the following four forms: $\{c^{i,j}_{p-1}, c^{i,j}_{p}\}$, $\{d^{i,j}_{p-1}, d^{i,j}_{p}\}$, $\{a^{i,j}_{p-1}, a^{i,j}_{p}\}$, and $\{b^{i,j}_{p-1}, b^{i,j}_{p}\}$.
\end{lemma}

		\begin{proof}
  By Lemma \ref{lTwoEdgesInVSG} and \ref{lFourEdgesInICG}, $F $ contains at least two edges per vertex-selection gadget and four edges per incidence-checking gadget.
  Since $\beta = 2 \#\text{vertex-selection gadgets} + 4\#\text{incidence-checking gadgets}$, it follows that $F$ contains exactly two edges from each vertex-selection gadget and exactly four edges from each incidence-checking gagdet.

		  Thus, by Lemma \ref{lTwoEdgesInVSG}, it follows that $F$ contains one edge from the upper and one edge from the lower path of each vertex-selection gadget.
  By Lemma \ref{lFourEdgesInICG}, we know that $F$ contains exactly one edge of each of the forms $\{a^{i,j}_{p-1}, a^{i,j}_{p}\}$,  $\{b^{i,j}_{p-1}, b^{i,j}_{p}\}$, $\{c^{i,j}_{p-1}, c^{i,j}_{p}\}$, and $\{d^{i,j}_{p-1}, d^{i,j}_{p}\}$ for each pair $(i, j) $ with $i<j$.
\end{proof}

		By Lemma \ref{loneedgeperpath}, we know that any $\lambda$-cut $F$ with $|F|\le \beta$ has to contain an edge from the upper and one edge from the lower path of any vertex selection gadget.
We now show that these edges indeed have to be at the same position, i.e., $\{u^i_{p-1}, u^i_p\} \in F$ if and only if~$\{\ell^i_{p-1}, \ell^i_p\}\in F$.

		\begin{lemma}\label{lvertexgadget}
  Let $F$ be an $\lambda$-cut with $|F|\le \beta$, and let $j\in [k]$.
		  Then, there exists an $x\in [n]$ such that a shortest $s$-$u^j_n$-path in $G- F$ has length at most~${2\eta + n + x + 1}$ and a shortest $s$-$\ell^j_n$-path in~$G- F$ has length at most $3 \eta + n -x + 1$.
\end{lemma}

		\begin{proof}
  By Lemma \ref{loneedgeperpath}, we know that $F$ contains exactly one edge of the form $\{u^j_x, u^j_{x+1}\}$ and one edge of the form $\{\ell^j_y, \ell^j_{y+1}\}$ from the vertex-selection gadget $A_i$.
		  Then $s$-$u^j_0$-$U^j_x$-$u^j_n$ is an $s$-$u^j_n$-path of length $2\eta +x + n+1$.

  If $y\ge x$, then $s$-$\ell^j_0$-$\ell^j_n$ is an $s$-$\ell^j_n$-path of length $\eta + 2\eta -y + n \le 3 \eta +n - x$, and the lemma follows.
  Otherwise $s$-$u^j_0$-$u^j_x$-$\ell^j_x$-$\ell^j_n$-$T^{j}_1$-$t$ is an $s$-$t$-path of length $2 \eta + n + 2 + \lambda - (2\eta + n - 2) = \lambda$, contradicting the assumption that~$F$ is a $\lambda$-cut.
\end{proof}

		If an $\lambda$-cut contains an edge $\{a^{i,j}_{p-1}, a^{i,j}_p\}$, then we say that the vertex-selection gadget selects vertex $\vertex (p)$.
We now show that if a vertex~$x$ is selected by a vertex-selection gadget~$A_i$, then~$F$ contains the edges $\{a^{i,j}_{\idx (x) -1}, a^{i,j}_{\idx (x)}\}$ and $\{b^{i,j}_{\idx (x) -1}, b^{i,j}_{\idx (x)}\}$ for all incidence checking gadgets $I^{i,j}$.

		\begin{lemma}\label{lupperpath}
  Consider an incidence-checking gadget $I^{i,j}$ and an $\lambda$-cut $F$ with $|F|\le \beta$.
  Assume that a shortest $s$-$u^i_n$-path in $H-F$ is of length $2\eta + v + n + 1$, and the length of a shortest~$s$-$\ell^i_n$-path in $H-F$ is at most $3\eta -v +n + 1$.
  Then,~$\{a^{i,j}_{v-1}, a^{i,j}_v\}\in F$ and~$\{b^{i,j}_{v-1}, b^{i,j}_v\}\in F$.
\end{lemma}

		\begin{proof}
  By Lemma \ref{loneedgeperpath}, $F$ contains an edge of the form $\{b^{i,j}_{p-1}, b^{i,j}_p\}$, and an edge of the form~$\{a^{i,j}_{q-1}, a^{i,j}_q\}$.
  We need to show $v= p= q$.

		  If $v > p$, then the path $s$-$\ell^i_0$-$\ell^i_n$-$b^{i,j}_0$-$B^{i,j}_p$-$b^{i,j}_n$-$t$ is of length at most $(3\eta - v + n + 1) + (n + 1) + 2\eta + p + 3\eta = 8\eta  + 2n + p -v + 2\le \lambda$, so we have $v\le p$.
  If $F$ contains the edge~$\{a^{i,j}_{q-1}, a^{i,j}_q\}$ for $q < p$, then the path $s$-$\ell^i_0$-$\ell^i_n$-$b^{i,j}_0$-$b^{i,j}_{p-1}$-$a^{i,j}_{p-1}$-$a^{i,j}_n$-$t$ is of length $(3\eta - v + n + 1) + n + 3 + 4 \eta = 7\eta + 2n -v + 4 < \lambda$ as $\eta\ge 3n$, contradicting the assumption that $F$ is an $\lambda$-cut.
  If we have $\{a^{i,j}_{q-1}, a^{i,j}_q\}\in F$ for $q > v$, then the path $s$-$u^i_0$-$u^i_n$-$a^{i,j}_0$-$a^{i,j}_n$-$t$ has length $(2\eta + v + n + 1) + 4\eta + n + 2\eta -q  + 1 < 8 \eta + 2n + 2 = \lambda + 1$, yielding an $s$-$t$-path of length at most $\lambda$.
  Thus, we have $v \le p \le q \le v$, proving the lemma.
\end{proof}

		We are now ready to show that the incidence-checking gadgets work as desired.
By Lemma \ref{lupperpath}, we know that any solution $F$ contains the edges $\{a^{i,j}_{v-1}, a^{i,j}_v\}$ and $\{b^{i,j}_{v-1}, b^{i,j}_v\}$ for some $p\in [n]$.
We now show that $F$ also has to contain the edges $\{c^{i,j}_{z-1}, c^{i,j}_z\}$ and $\{c^{i,j}_{z-1}, c^{i,j}_z\}$, where $e_z$ is the edge between the vertices selected by $A_i$ and $A_j$.
This shows in particular that such an edge exists, and therefore, a solution to the \textsc{Length-Bounded Cut} instance corresponds to a clique of size~$k$ in~$G$.

		\begin{lemma}
  \label{llowerpath}
  Consider an incidence-checking gadget $I^{i,j}$ and an $\lambda$-cut $F$ with $|F|\le \beta$.
  Assume that~$F$ contains the edges $\{a^{i,j}_{p-1}, a^{i,j}_p\}$ and $\{b^{i,j}_{p-1}, b^{i,j}_p\}$ for some $p\in [n]$.
  Let~$q\in [m]$ such that~$H$ contains an $a^{i,j}_{p-1}$-$c^{i,j}_q$-path of length $2 \eta$ and $r\in [q, m]$ such that $H$ contains an~$a^{i,j}_{p}$-$ c^{i,j}_r$-path of length $2\eta$.

		  Then, there exists a $z\in [q + 1, r]$ such that $F$ contains the edges $\{c^{i,j}_{z-1}, c^{i,j}_z\} $ and~$\{d^{i,j}_{z-1}, d^{i,j}_z\}$.
  Moreover, if $x$ is the vertex selected by the preceding vertex-selection gadget~$A_j$, then $e_z = \{\vertex (p), x\}$ and $p < \idx (x)$.
\end{lemma}

		\begin{proof}
  By Lemma \ref{loneedgeperpath}, we know that $F$ contains exactly one edge of the form $\{c^{i,j}_{\gamma-1}, c^{i,j}_\gamma\}$ and one edge of the form $\{d^{i,j}_{\delta-1}, d^{i,j}_\delta\}$.

		  We first show that $\gamma, \delta\in [q+1, r]$, implying that the edges with index $\gamma$ and $\delta$ are incident to $\vertex (p)$.

\noindent{\bfseries Case 1 ($\gamma < q + 1$):}
  Then $s$-$\ell^i_0$-$L^i_p$-$\ell^i_n$-$b^{i,j}_0$-$b^{i,j}_{p-1}$-$c^{i,j}_q$-$c^{i,j}_m$-$t$ is an $s$-$t$-path of length $(\eta + 2) + (2\eta - p) + (n - 1) + 2 + p + 3 \eta + (m - q) + (\eta + n - m + 2) = 7\eta + 2n + 5 - q < \lambda$.

\noindent{\bfseries Case 2 ($\delta < q + 1$):}
  Similarly to Case 1, $s$-$\ell^i_0$-$L^i_p$-$\ell^i_n$-$b^{i,j}_0$-$b^{i,j}_{p-1}$-$d^{i,j}_q$-$d^{i,j}_m$-$t$ is an $s$-$t$-path of length $(\eta + 2) + (2\eta - p) + (n - 1) + 2 + p + 2 \eta + (m - q) + (2\eta + n - m + 2) = 7\eta + 2n + 5 - q < \lambda$.

\noindent{\bfseries Case 3 ($\gamma > r$):}
  Then $s$-$u^j_0$-$U^j_{\idx (x)}$-$u^j_n$-$c^{i,j}_0$-$c^{i,j}_r$-$a^{i,j}_p$-$a^{i,j}_n$-$t$ is an $s$-$t$-path of length $2 + (n-1) + (2\eta + x) + 3\eta + r + 2\eta + (n-p) + 2 = 7\eta + 2n + x - p + r + 3 < \lambda$.

\noindent{\bfseries Case 4 ($\delta > r$):}
  Similarly to Case 3, $s$-$\ell^j_0$-$L^j_{\idx (x)}$-$\ell^j_n$-$d^{i,j}_0$-$d^{i,j}_r$-$a^{i,j}_p$-$a^{i,j}_n$-$t$ is an $s$-$t$-path of length at most $(\eta + 2) + (n -1) + (2\eta + x) + \eta + r + 3\eta + (n-p) + 2 = 7\eta + 2n + x - p + r + 3< \lambda$.

		  It remains to show that $\gamma = \delta$, and $e_\gamma = \{\vertex (p), x\}$ with $p < \idx (x)$.
  By Lemma~\ref{lvertexgadget}, we know that a shortest $s$-$u^j_n$-path has length at most $2\eta + n +\idx ( x) + 1$ and that the length of a shortest $s$-$\ell^j_n$-path has length at most $3\eta + n- \idx (x) + 1$.

		  By the choice of $q$ and $r$, we have that $e_\alpha$ is incident to $\vertex (p)$ for $\alpha \in [q+1, r]$.
  Denote by $u_\alpha$ the index of the other endpoint of $e_\alpha$.
  Note that $u_\alpha > p$ by the construction of $H$.
  We will show that $\idx (x) \le u_\delta\le u_\gamma \le \idx (x)$.

		  If $\idx (x) > u_\delta$, then $s$-$\ell^j_0$-$L^j_{\idx (x)}$-$\ell^j_n$-$d^{i,j}_0$-$D^{i,j}_{\delta}$-$d^{i,j}_m$-$t$ is an $s$-$t$-path of length at most $(3\eta + n - \idx (x) + 1) + \eta + (2\eta + u_\delta) + (m - 1)  + (2\eta + n - m + 2) = 8\eta + 2n + u_\delta -\idx (x) + 2\le \lambda$, a contradiction to $F$ being a solution.
  Thus, we have $\idx (x) \le u_\delta$.
  If $u_\delta > u_\gamma$, then $s$-$\ell^j_0$-$L^j_{\idx (x)}$-$\ell^j_n$-$d^{i,j}_0$-$d^{i,j}_{\delta-1}$-$c^{i,j}_{\delta-1}$-$c^{i,j}_m$-$t$ is an $s$-$t$-path of length at most $(3\eta + n - \idx (x) + 1) + \eta + m + 2\eta + (\eta+ n - m +2) \le 7\eta + 2n + 3 < \lambda$, a contradiction to $F$ being a solution.
  Thus, we have $u_\delta \le u_\gamma$.
  If $u_\gamma > x$, then $s$-$u^j_0$-$U^j_{\idx (x)}$-$u^j_n$-$c^{i,j}_0$-$C^{i,j}_\gamma$-$c^{i,j}_m$-$t$ is an $s$-$t$-path of length $(2\eta + n + \idx (x) + 1) + 3\eta + (m - 1) + (2\eta - u_\gamma) +(\eta +n -m + 2) = 8 \eta + 2n + \idx (x) - u_\gamma + 2 = \lambda +1 + \idx (x) - u_\gamma \le \lambda$, a contradiction to $F$ being a solution.
  Thus, we have $u_\gamma \le x$.

		  Thus, we have $\idx (x) \le u_\delta \le u_\gamma \le \idx (x)$.
  By the construction of $H$, the edge $e_\gamma$ corresponds to an edge $\{\vertex (p), x\} $ with $p < \idx (x)$, proving the lemma.
\end{proof}

		With this at hand, we are now ready to prove the forward direction.

		\begin{lemma}\label{lforward}
  If there is an $\lambda$-cut of size (at most) $\beta$ in $H$, then there is a clique of size $k$ in $G$.
\end{lemma}

		\begin{proof}
  By Lemma \ref{lvertexgadget}, each vertex-selection gadget $A_i$ selects a vertex $x_i$.
  We claim that~$\{x_1, \dots, x_k\}$ is a clique.
  Consider two vertices $x_i $ and $x_j$ with $i< j$.
  By Lemma \ref{lupperpath}, we know that $F$ contains the edges $\{a^{i,j}_{\idx (c_i)-1}, a^{i,j}_{\idx (c_i)}\}$ and~$\{b^{i,j}_{\idx(c_i)-1}, b^{i,j}_{\idx (c_i)}\}$.
  Let $q\in [m]$ such that $H$ contains an~$a^{i,j}_{\idx (c_i) -1}$-$c^{i,j}_q$-path of length~$2 \eta$, and let $r\in [m] $ such that $H$ contains an~\mbox{$a^{i,j}_{\idx (c_i)}$-$c^{i,j}_q$-path} of length~$2 \eta$.
  By Lemma \ref{llowerpath}, the cut~$F$ contains the edges $\{c^{i,j}_{z-1}, c^{i,j}_z\}$ and $\{d^{i,j}_{z-1}, d^{i,j}_z\}$ for some~${z\in [q+1, r]}$ with~$e_z  = \{c_i, c_j\}$.
  Thus, by the construction of $H$, we know that~${\{c_i, c_j\}\in E}$.\qedhere
\end{proof}

		\subsection{Backward Direction}

		Starting with a clique $C = \{x_1, \dots, x_k\}$ in $G$, we construct a $\lambda$-cut $F$.
In this cut $F$, each vertex-selection gadget will select one vertex from the clique, and $F$ will contain edges in the incidence-checking gadget $I^{i,j}$ corresponding to the edge $\{x_i, x_j\}$.
We then have to show that~$F$ is indeed a $\lambda$-cut.
In order to do so, we will first evaluate the length of a shortest $s$-$u^i_n$- and a shortest~$s$-$\ell^i_n$-path, and afterwards compute the length of a shortest $u^i_n$-$t$- and~$\ell^i_n$-$t$-path.

		\begin{lemma}\label{lnoshortvertexpath}
  Consider a vertex-selection gadget.
  After deleting $\{u^j_{p-1}, u^j_{p}\}$ and $\{\ell^j_{p-1}, \ell^j_p\}$, any $s$-$u^j_n$-path has length at least $2\eta + p + n + 1$, and any $s$-$\ell^j_n$-path has length at least $3\eta - p + n + 1$.
\end{lemma}

		\begin{proof}
  The connectivity paths have length $\lambda - (4\eta + n + 2) = 4\eta +n - 1$ or $\lambda -(3\eta +n + 2) = 5\eta + n -1$ and are therefore longer.

		  All other $s$-$u^j_n$-paths or $s$-$\ell^j_n$-paths contain the path $L^{j}_p$ of length $2\eta -p$ between $\ell^j_{p-1}$ and~$\ell^j_p$, or the path $U^j_p$ of length $2\eta + p$ between $u^j_{p-1} $ and $u^j_p$.
  Thus, any path containing one of the $u^j_q$-$\ell^j_q$-paths has length at least $2\eta + 2\eta - n > 3\eta - p + n + 1$, so we may assume that a shortest $s$-$u^j_n$- or $s$-$\ell^j_n$-path does not contain such a $u^j_q$-$\ell^j_q$-path.

		  The path $s$-$u^j_0$-$u^j_1$-$U^j_p$-$u^j_n$ has length $2\eta + p + n + 1$, and the path $s$-$\ell^j_0$-$L^j_p$-$\ell^j_n$ has length $(\eta+2) + (n-1) + (2\eta - p) = 3\eta - p + n + 1$.
\end{proof}

		We now compute the length of a shortest $u^i_n$-$t$- and $\ell^i_n$-$t$-path.

		\begin{lemma}\label{lnoshortincidencepath}
  Consider an incidence-checking gadget $I^{i,j}$.
  Let $p\in [n]$, let $a^{i,j}_{p-1}$ be connected to $c^{i,j}_q$ by a path of length $2 \eta$, and let $a^{i,j}_p$ be connected to $c^{i,j}_{q'} $ by a path of length $2 \eta$.
  After deleting $\{a^{i,j}_{p-1}, a^{i,j}_p\}$, $\{b^{i,j}_{p-1}, b^{i,j}_p\}$, $\{c^{i,j}_{z-1}, c^{i,j}_z\}$, and $\{d^{i,j}_{z-1}, d^{i,j}_z\}$ for some $z\in [q + 1, q']$,
  \begin{itemize}
    \item any $u^i_n$-$t$-path has length at least $6\eta + n - p + 1$,
    \item any $\ell^i_n$-$t$-path has length at least $5\eta + n + p + 1$,
    \item any $u^j_n$-$t$-path has length at least $6\eta + n - u_z + 1$,
    \item any $\ell^j_n$-$t$-path has length at least $5\eta + n + u_z + 1$.
  \end{itemize}
\end{lemma}

		\begin{proof}
  First note that any path starting at $u^i_n$ or $u^j_n$ accumulates at least $4\eta $ edges, while any path starting at $\ell^i_n$ or $\ell^j_n$ accumulates at least $3\eta$ edges.

		  Furthermore, since $z\in [q+1, q']$, any path contains a path of length at least $2\eta -n$ between $a^{i,j}_{p-1}$ and $a^{i,j}_p$, $b^{i,j}_{p-1}$ and $b^{i,j}_p$, $c^{i,j}_{z-1}$ and $c^{i,j}_z$ or $d^{i,j}_{z-1}$ and $d^{i,j}_z$.

		  Any path using vertices of the form $x^{i,j}_\alpha$ for two different $x\in \{a,b,c,d\}$ is then of length at least $7\eta - n$ if it starts an $u^i_n$ or $u^j_n$, and at least $6\eta - n$ if it starts in $\ell^i_n$ or $\ell^j_n$.
  Thus, $u^i_n$-$t$-, $\ell^i_n$-$t$-, $u^j_n$-$t$-, and $\ell^j_n$-$t$-paths containing vertices of the form $x^{i,j}_\alpha$ for two different $x\in \{a,b,c,d\}$ are longer than the lower bound claimed by the lemma.

		  Any connectivity path starting in $u^i_n$ or $u^j_n$ has length $\lambda-(n+2) = 8\eta + n - 1$ is longer than $6\eta + n - p + 1$; similarly, any connectivity path starting in $\ell^i_n$ or $\ell^j_n$ has length $\lambda - (\eta + 2 + n) = 7\eta+ n - 1$ and therefore is longer than $5\eta + n + \max\{p, u_z\} + 1$.

		  It remains to check four paths: $u^i_n$-$a^{i,j}_0$-$A^{i,j}_p$-$a^{i,j}_n$-$t$ (of length $6\eta - p + n + 1$), $\ell^i_n$-$b^{i,j}_0$-$B^{i,j}_p$-$b^{i,j}_n$-$t$ (of length $n+ 5\eta + p + 1$), $u^j_n$-$c^{i,j}_0$-$C^{i,j}_z$-$c^{i,j}_m$-$t$ (of length $3\eta + (2\eta - u_i) + m -1 + (\eta+ n -m +2) =  6\eta + n -u_i + 1$) and $\ell^j_n$-$d^{i,j}_0$-$D^{i,j}_z$-$d^{i,j}_m$-$t$ (of length $\eta + (2\eta + u_i) + m-1 + (2\eta + n - m + 2) = 5\eta + n -u_i + 1$).
\end{proof}

		Now the backward direction easily follows.

		\begin{lemma}\label{lbackward}
  If $G$ contains a clique of size $k$, then $H$ contains an $\lambda$-cut of size $\beta$.
\end{lemma}

		\begin{proof}
  Let $\{x_1, \dots, x_k\}$ be a clique with $\idx (x_i) < \idx (x_{i+1})$.
  The $\lambda$-cut consists of the $j$-th vertex-selection gadget $A_j$ selecting $x_j$ (i.e., $\{u^j_{\idx(x_j)-1}, u^j_{\idx (x_j)}\}\in F$ and $\{\ell^j_{\idx (x_j)-1}, \ell^j_{\idx (x_j)}\}\in F$), and the incidence gadget between two vertex-selection gadgets~$A_i$ and $A_j$ with $i< j$ the corresponding edges (i.e., $\{a^{i,j}_{\idx (x_i)-1}, a^{i,j}_{\idx (x_i)}\}\in~F$, \allowbreak$\{b^{i,j}_{\idx (x_i)-1}, b^{i,j}_{\idx(x_i)}\}\in~F$, and $\{c^{i,j}_{p -1}\}, c^{i,j}_p\}\in F$ and $\{d^{i,j}_{p -1}, d^{i,j}_p\}\in F$, where $e_p = \{x_i, x_j\}$).
  Notice that $F$ contains exactly~$\beta$ edges.

		  It remains to show that any $s$-$t$-path has length at least $\lambda + 1$.
  Combining Lemmata \ref{lnoshortvertexpath} and \ref{lnoshortincidencepath}, we get that a shortest $s$-$t$-path passing through $v_n$ has length at least $(2\eta + \idx (x_i) + n + 1) + (6\eta + n - \idx(x_i) + 1) = 8\eta + 2n + 2 = \lambda + 1$, and a shortest $s$-$t$-path passing through $\ell^j_n$ has length at least $(3\eta + n - \idx (x_i) + 1) + (5\eta + n + \idx (x_i) + 1) = 8\eta + 2n + 2 = \lambda + 1$.
\end{proof}

		\subsection{Running Time and Parameter Size}
We now analyze the running time of our reduction from \textsc{Clique} and bound the pathwidth and the maximum degree in the resulting instance in the original parameter~$k$.
We first show  that the construction takes~$O(k^2 \cdot m^2)$ time.

		\begin{lemma}
\label{lem:pwmxrunningtime}
Given an instance~$(G,k)$ of \textsc{Clique} parameterized by solution size, the resulting instance~$(H,s,t,\beta,\lambda)$ of \lbc{} can be computed in~$O(k^2 \cdot m^2)$ time.
\end{lemma}

		\begin{proof}
First, observe that the computation for $\eta,\lambda,$ and~$\beta$ take constant time once~$n,m,$ and~$k$ are known.
The construction of the~$k$ vertex-selection gadgets take~$O(n \cdot m)$ time each, resulting in an overall running time in~$O(k\cdot (n + m))$.
For each of the~$\binom{k}{2}$ incidence-checking gadgets, it takes~$O(m^2)$ time to construct the the gadget, resulting in overall~$O(k^2 \cdot m^2)$ time.
Finally there are~$O(k^2)$ connectivity paths, each of which having length in~$O(m)$.
Since each of these paths can be constructed in time linear in its length, the time for constructing all paths is bounded by~$O(k^2 \cdot m)$.
Concluding, the time to compute the reduction is bounded by~$O(k^2 \cdot m^2) \subseteq O(n^2 \cdot m^2)$, which is clearly polynomial in the input size.
\end{proof}

		Next we analyze the pathwidth of the constructed graph~$H$.
We start with a helpful lemma for suppressing degree-two vertices.

		\begin{lemma}\label{lpw}
  Let $G$ be a graph.
  Successively suppressing vertices of degree two increases the pathwidth by at most two.
\end{lemma}

		\begin{proof}
  Consider an optimal path decomposition, and let $P$ be a suppressed path with end vertices~$v$ and $w$.
  Then there is a bag containing both $v$ and $w$.
  We double that bag, and insert between the two copies of $v$ and $w$ the trivial path decomposition of $P$ of width one.
  This results in a path decomposition of the suppressed graph, and the width of the decomposition increased by at most two.
\end{proof}

		With this at hand, we next analyze the pathwidth of the different gadgets, starting with vertex-selection gadgets.

		\begin{lemma}\label{lpwvsg}
  The pathwidth of a vertex-selection gadget is at most five.
\end{lemma}

		\begin{proof}
  We give a path decomposition of this gadget as follows:
  We first successively suppress vertices of degree two.
  By Lemma~\ref{lpw}, this increases the pathwidth by at most two.

		  The $i$-th bag contains the vertices $\{u^j_{i-1}, \ell^j_{i-1}, u^j_i, \ell^j_i\}$.

		  It is easy to check that this is indeed a path decomposition of a vertex-selection gadget.
\end{proof}

		We now analyze the pathwidth of incidence-checking gadgets. Notice that the pathwidth of connectivity paths is one as they are paths.

		\begin{lemma}\label{lpwicg}
  The pathwidth of an incidence-checking gadget is at most nine.
\end{lemma}

		\begin{proof}
  We give a path decomposition of the incidence-checking gadget $I^{i,j}$ as follows:
  We first successively suppress vertices of degree two.
  By Lemma~\ref{lpw}, this increases the pathwidth by at most two.

		  Each bag will contain eight vertices $a^{i,j}_p$, $b^{i,j}_p$, $a^{i,j}_{p+1}$, $b^{i,j}_{p+1}$, $c^{i,j}_q$, $d^{i,j}_q$, $c^{i,j}_{q+1}$, and~$d^{i,j}_{q+1}$.
  The first bag of the path decomposition contains the vertices $a^{i,j}_0$, $b^{i,j}_0$, $a^{i,j}_{1}$, $b^{i,j}_{1}$, $c^{i,j}_0$, $d^{i,j}_0$,~$c^{i,j}_{1}$, and~$d^{i,j}_{1}$.

		  Now let the $\alpha$-th bag of the path decomposition contain $a^{i,j}_p$, $b^{i,j}_p$, $a^{i,j}_{p+1}$, $b^{i,j}_{p+1}$, $c^{i,j}_q$, $d^{i,j}_q$,~$c^{i,j}_{q+1}$, and $d^{i,j}_{q+1}$.
  By the construction of $H$, there exists a unique $\beta \in [m]$ such that the vertex~$a^{i,j}_{p+1}$ is connected to the vertex $c^{i,j}_{\beta}$ by a path of length $2\eta$.
  We construct the $\alpha + 1$-th bag as follows.
  If~$q < \beta$, then we add $d^{i,j}_{q + 2}$ and $c^{i,j}_{q+2}$ and remove $d^{i,j}_{q}$ and $c^{i,j}_{q}$;
  otherwise we add~$a^{i,j}_{p +2}$ and~$b^{i,j}_{p + 2}$ and remove $a^{i,j}_{p}$ and $b^{i,j}_{p}$.

		  It it easy to check that this is indeed a path decomposition of the graph arising from successively suppressing degree-two vertices from the incidence-checking gadget $I^{i,j}$, and it clearly has width seven.
  Due to the suppressing of degree-two vertices, the pathwidth can increase by two, so the pathwidth of an incidence-checking gadget is at most nine.
\end{proof}

		We can now combine \Cref{lpwvsg,lpwicg} to analyze the pathwidth of~$H$ to be in~$O(k)$.

		\begin{lemma}\label{lpathwidth}
  The pathwidth of $H$ is bounded by $O(k)$.
\end{lemma}

		\begin{proof}
  After removing $s$, $t$, and the $2k$ vertices $\{u^i_n, \ell^i_n: i\in [n]\}$ from $H$, the connected components of the resulting graph consist of vertex selection gadgets (of pathwidth at most five (by Lemma~\ref{lpwvsg}), incidence-checking gadgets (of pathwidth is at most nine by Lemma~\ref{lpwicg}), and connectivity paths (of pathwidth one, as they are paths).

		  Thus, the pathwidth of $H$ is at most $(2k+ 2) + 9= 2k+11$.
\end{proof}

		Finally, we analyze the maximum degree in~$H$. Observe that each inner vertex in any gadget has only constantly many incident edges and all vertices in~$\{s,t\}\cup\{u^i_n, \ell^i_n: i\in [n]\}$ have~$O(k^2)$ incident edges. Hence the maximum degree in~$H$ is in~$O(k^2)$.

		\begin{observation}\label{odegree}
  The maximum degree is $O(k^2)$.
\end{observation}

		Combining all of these results, it is easy to see that

		\pwmdW*

		\begin{proof}
  Immediately follows from Lemmata~\ref{lforward},~\ref{lbackward},~\ref{lem:pwmxrunningtime},~\ref{lpathwidth} and Observation~\ref{odegree}.
\end{proof}

\section{W[1]-Hardness for Feedback Vertex Number}
In this section we will prove that \lbc{} is W[1]-hard with respect to the feedback vertex number.
We will present a parameterized reduction from \textsc{Multicolored Clique} parameterized by solution size.
\textsc{Multicolored Clique} is known to be W[1]-hard with respect to the solution size and is formally defined as follows.
\defProblemQuestion{\textsc{Multicolored Clique}}
{A $k$-partite undirected graph $G=(V,E)$ with $V = \bigcup_{i=1}^k V_i$ and $|V_i| = \otherN$ for all~$i\in [k]$.}
{Does $G$ contain a clique of size $k$?}

After we present the reduction, we show its correctness and finally analyze its running time and the size of the feedback vertex number in the resulting instance.

\subsection{The Reduction}
We will describe our reduction for a given instance~$(G=(V,E),k)$ of \textsc{Multicolored Clique} and call the resulting graph $H= (V_H, E_H)$.
We assume that we are given a function $\idx \colon V \rightarrow [\otherN]$ and for each $i \in [k]$, a function $\vertex_i \colon [\otherN] \rightarrow V_i$ such that $\vertex_i (\idx (v) ) = v$ for every~$v \in V_i$ and $\idx (\vertex_i (x) ) = x$ for all $x\in [\otherN]$.

We start by adding for each $V_i$ a specific vertex-selection gadget that is explained below.
Second, we add for each~$i \in [k]$, each~$j \in [k]\setminus \{i\}$, and each edge between $V_i$ and $V_j$ an edge gadget (which is also described below).
We then set $\lambda := \otherN + 2\tbdd$ and $\beta:= 2k (\otherN - 1) \tbd + m - \binom{k}{2}$, where~$\tbdd$~($m$) is the number of vertices (edges) of~$G$ (i.e., $\tbdd = k\otherN$).

Inside each vertex selection gadget, any $\lambda$-cut $F$ of size $\beta$ will contain $2 (\otherN -1)$ edges, selecting a vertex from $V_i$.
Furthermore, any~$\lambda$-cut $F$ of size at most~$\beta$ will contain an edge from each edge gadget for which the corresponding edge is not between two selected vertices.

\subsubsection{Vertex-Selection Gadgets}

Each vertex selection gadget $A_i$ starts in $s$ and ends in $t$.
It has two ``middle vertices'' $u_i$ and~$\ell_i$.
Between $s$ and $u_i$ ($\ell_i$), there is a path $S_i^{j,p}$ ($\bar{S}_i^{j,p}$) of length $\tbdd + j$ for each $1 \le j\le \otherN$ and $1\le p \le \tbd$, and a path $T_i^{j,p}$ ($\bar{T}_i^{j,p}$) of length $\tbdd + j$ between $u_i$ ($\ell_i$) and $t$ for each $1\le j\le \otherN$ and $1\le p \le \tbd$.

Finally, we add ``shortcut edges''.
From the second vertex $s^{j,p}_i$ (i.e., the vertex adjacent to $s$) of each path $S^{j,p}_i$, we add an edge $c^{j,p}_i$ to the second-last vertex $\bar{s}^{\otherN- j, p}_i$ (i.e. the vertex adjacent to~$\ell_i$) of $\bar{S}^{\otherN - j, p}_i$ for every $1\le j < \otherN$ and $p\in [\tbd]$.
Similarly, from the second vertex~$\bar{t}^{\otherN -j,p}_i$ (i.e., the vertex adjacent to $\ell_i$) of the path $\bar{T}^{\otherN-j, p}$, we add an edge $\bar{c}^{j,p}_i$ to the second-last vertex $t^{j, p}_i$ (i.e., the vertex adjacent to $t$) of $T^{j, p}_i$ for every $1\le j < \otherN$ and $p\in [\tbd]$.
An example of a vertex-selection gadget can be seen in Figure~\ref{fVertexSelectionGadgetFVS}.

		\begin{figure}
		  \begin{center}
		    \begin{tikzpicture}[xscale =0.7]
		      \node[vertex, label=180:$s$] (s) at (0, -7.) {};
		      \node[vertex, label=0:$t$] (t) at (20, -7) {};

		      \node[vertex, label=90:$u_i$] (ui) at (10, 0) {};
		      \node[vertex, label=270:$\ell_i$] (li) at (10, -15) {};

		      \node[vertex, label=90:$s^{1,1}_i$] (s1) at (5, 2) {};
		      \node[vertex, label=90:$s^{2,1}_i$] (s21) at (4, 1) {};
		      \node[vertex] (s22) at (6, 1) {};
		      \node[vertex, label=90:$s^{3,1}_i$] (s31) at (3, 0) {};
		      \node[vertex] (s32) at (5, 0) {};
		      \node[vertex] (s33) at (7, 0) {};
		      \node[vertex, label=90:$s^{4,1}_i$] (s41) at (2, -1) {};
		      \node[vertex] (s42) at (4, -1) {};
		      \node[vertex] (s43) at (6, -1) {};
		      \node[vertex] (s44) at (8, -1) {};

		      \begin{scope}[yshift = -10cm]

		      \node[vertex, label=90:$\bar{s}^{1,1}_i$] (sb1) at (5, 2) {};
		      \node[vertex] (sb21) at (4, 1) {};
		      \node[vertex, label=90:$\bar{s}^{2,1}_i$] (sb22) at (6, 1) {};
		      \node[vertex] (sb31) at (3, 0) {};
		      \node[vertex] (sb32) at (5, 0) {};
		      \node[vertex, label=90:$\bar{s}^{3,1}_i$] (sb33) at (7, 0) {};
		      \node[vertex] (sb41) at (2, -1) {};
		      \node[vertex] (sb42) at (4, -1) {};
		      \node[vertex] (sb43) at (6, -1) {};
		      \node[vertex, label=90:$\bar{s}^{4,1}_i$] (sb44) at (8, -1) {};

		      \end{scope}

		      \draw (s) --(s1);
		      \draw (s1) edge node[pos=0.5, fill=white, inner sep=1pt] {\scriptsize $\tbdd$} (ui);
		      \draw (ui) edge node[pos=0.5, fill=white, inner sep=1pt] {\scriptsize $\tbdd$} (s22);
		      \draw (s22) -- (s21) -- (s) -- (s31) -- (s32) --(s33);
		      \draw (s33) edge node[pos=0.5, fill=white, inner sep=1pt] {\scriptsize $\tbdd$} (ui);
		      \draw (ui) edge node[pos=0.5, fill=white, inner sep=1pt] {\scriptsize $\tbdd$} (s44);
		      \draw (s44) --(s43) -- (s42) -- (s41);
		      \draw (s41) -- (s);
		      \draw (s) edge node[pos=0.5, fill=white, inner sep=1pt] {\scriptsize $\tbdd$} (sb1);
		      \draw (sb1) -- (li) -- (sb22) -- (sb21);
		      \draw (sb21) edge node[pos=0.5, fill=white, inner sep=1pt] {\scriptsize $\tbdd$} (s);
		      \draw (s) edge node[pos=0.5, fill=white, inner sep=1pt] {\scriptsize $\tbdd$} (sb31);
		      \draw (sb31) -- (sb32) --(sb33) -- (li) -- (sb44) --(sb43) -- (sb42) -- (sb41);
		      \draw (sb41) edge node[pos=0.5, fill=white, inner sep=1pt] {\scriptsize $\tbdd$} (s);

		      \draw (s1) edge[blue, dashed] node[pos=0.5, label=0:$c^{1, 1}_i$] {} (sb33);
		      \draw (s21) edge[blue, dashed] node[pos=0.5, label=0:$c^{2, 1}_i$] {} (sb22);
		      \draw (s31) edge[blue, dashed] node[pos=0.5, label=180:$c^{3, 1}_i$] {} (sb1);

		      \begin{scope}[xshift = 10 cm]

		      \node[vertex, label=90:$t^{1,1}_i$] (s1) at (5, 2) {};
		      \node[vertex] (s21) at (4, 1) {};
		      \node[vertex, label=90:$t^{2,1}_i$] (s22) at (6, 1) {};
		      \node[vertex] (s31) at (3, 0) {};
		      \node[vertex] (s32) at (5, 0) {};
		      \node[vertex, label=90:$t^{3,1}_i$] (s33) at (7, 0) {};
		      \node[vertex] (s41) at (2, -1) {};
		      \node[vertex] (s42) at (4, -1) {};
		      \node[vertex] (s43) at (6, -1) {};
		      \node[vertex, label=90:$t^{4,1}_i$] (s44) at (8, -1) {};

		      \begin{scope}[yshift = -10cm]

		      \node[vertex, label=90:$\bar{t}^{1,1}_i$] (sb1) at (5, 2) {};
		      \node[vertex, label=90:$\bar{t}^{2,1}_i$] (sb21) at (4, 1) {};
		      \node[vertex] (sb22) at (6, 1) {};
		      \node[vertex, label=90:$\bar{t}^{3,1}_i$] (sb31) at (3, 0) {};
		      \node[vertex] (sb32) at (5, 0) {};
		      \node[vertex] (sb33) at (7, 0) {};
		      \node[vertex, label=90:$\bar{t}^{4,1}_i$] (sb41) at (2, -1) {};
		      \node[vertex] (sb42) at (4, -1) {};
		      \node[vertex] (sb43) at (6, -1) {};
		      \node[vertex] (sb44) at (8, -1) {};

		      \draw (ui) edge node[pos=0.5, fill=white, inner sep=1pt] {\scriptsize $\tbdd$} (s1);
		      \draw (s1) -- (t) -- (s22) -- (s21);
		      \draw (s21) edge node[pos=0.5, fill=white, inner sep=1pt] {\scriptsize $\tbdd$} (ui);
		      \draw (ui) edge node[pos=0.5, fill=white, inner sep=1pt] {\scriptsize $\tbdd$} (s31);
		      \draw (s31) -- (s32) -- (s33) -- (t) -- (s44) -- (s43) -- (s42) -- (s41);
		      \draw (s41) edge node[pos=0.5, fill=white, inner sep=1pt] {\scriptsize $\tbdd$} (ui);

		      \draw (li) -- (sb1);
		      \draw (sb1) edge node[pos=0.5, fill=white, inner sep=1pt] {\scriptsize $\tbdd$} (t);
		      \draw (t) edge node[pos=0.5, fill=white, inner sep=1pt] {\scriptsize $\tbdd$} (sb22);
		      \draw (sb22) -- (sb21) -- (li) -- (sb31) -- (sb32) -- (sb33);
		      \draw (sb33) edge node[pos=0.5, fill=white, inner sep=1pt] {\scriptsize $\tbdd$} (t);
		      \draw (t) edge node[pos=0.5, fill=white, inner sep=1pt] {\scriptsize $\tbdd$} (sb44);
		      \draw (sb44) -- (sb43) -- (sb42) -- (sb41) -- (li);
		      \end{scope}

		      \draw (s1) edge[blue, dashed] node[pos=0.5, label=180:$\bar{c}^{1, 1}_i$] {} (sb31);
		      \draw (s22) edge[blue, dashed] node[pos=0.5, label=0:$\bar{c}^{2, 1}_i$] {} (sb21);
		      \draw (s33) edge[blue, dashed] node[pos=0.5, label=0:$\bar{c}^{3, 1}_i$] {} (sb1);
		      \end{scope}

		      \begin{scope}[yshift = - 5cm, color=gray]

		      \node[vertex, label=90:$s^{1,2}_i$] (s1) at (5, 2) {};
		      \node[vertex, label=90:$s^{2,2}_i$] (s21) at (4, 1) {};
		      \node[vertex] (s22) at (6, 1) {};
		      \node[vertex, label=90:$s^{3,2}_i$] (s31) at (3, 0) {};
		      \node[vertex] (s32) at (5, 0) {};
		      \node[vertex] (s33) at (7, 0) {};
		      \node[vertex, label=90:$s^{4,2}_i$] (s41) at (2, -1) {};
		      \node[vertex] (s42) at (4, -1) {};
		      \node[vertex] (s43) at (6, -1) {};
		      \node[vertex] (s44) at (8, -1) {};

		      \begin{scope}[yshift = -10cm]

		      \node[vertex, label=90:$\bar{s}^{1,2}_i$] (sb1) at (5, 2) {};
		      \node[vertex] (sb21) at (4, 1) {};
		      \node[vertex, label=90:$\bar{s}^{2,2}_i$] (sb22) at (6, 1) {};
		      \node[vertex] (sb31) at (3, 0) {};
		      \node[vertex] (sb32) at (5, 0) {};
		      \node[vertex, label=90:$\bar{s}^{3,2}_i$] (sb33) at (7, 0) {};
		      \node[vertex] (sb41) at (2, -1) {};
		      \node[vertex] (sb42) at (4, -1) {};
		      \node[vertex] (sb43) at (6, -1) {};
		      \node[vertex, label=90:$\bar{s}^{4,2}_i$] (sb44) at (8, -1) {};

		      \end{scope}

		      \draw (s) --(s1);
		      \draw (s1) edge node[pos=0.5, fill=white, inner sep=1pt] {\scriptsize $\tbdd$} (ui);
		      \draw (ui) edge node[pos=0.5, fill=white, inner sep=1pt] {\scriptsize $\tbdd$} (s22);
		      \draw (s22) -- (s21) -- (s) -- (s31) -- (s32) --(s33);
		      \draw (s33) edge node[pos=0.5, fill=white, inner sep=1pt] {\scriptsize $\tbdd$} (ui);
		      \draw (ui) edge node[pos=0.5, fill=white, inner sep=1pt] {\scriptsize $\tbdd$} (s44);
		      \draw (s44) --(s43) -- (s42) -- (s41);
		      \draw (s41) -- (s);
		      \draw (s) edge node[pos=0.5, fill=white, inner sep=1pt] {\scriptsize $\tbdd$} (sb1);
		      \draw (sb1) -- (li) -- (sb22) -- (sb21);
		      \draw (sb21) edge node[pos=0.5, fill=white, inner sep=1pt] {\scriptsize $\tbdd$} (s);
		      \draw (s) edge node[pos=0.5, fill=white, inner sep=1pt] {\scriptsize $\tbdd$} (sb31);
		      \draw (sb31) -- (sb32) --(sb33) -- (li) -- (sb44) --(sb43) -- (sb42) -- (sb41);
		      \draw (sb41) edge node[pos=0.5, fill=white, inner sep=1pt] {\scriptsize $\tbdd$} (s);

		      \draw (s1) edge[red, dashed] node[pos=0.5, label=0:$c^{1, 2}_i$] {} (sb33);
		      \draw (s21) edge[red, dashed] node[pos=0.5, label=0:$c^{2, 2}_i$] {} (sb22);
		      \draw (s31) edge[red, dashed] node[pos=0.5, label=180:$c^{3, 2}_i$] {} (sb1);

		      \begin{scope}[xshift = 10 cm]

		      \node[vertex, label=90:$t^{1,2}_i$] (s1) at (5, 2) {};
		      \node[vertex] (s21) at (4, 1) {};
		      \node[vertex, label=90:$t^{2,2}_i$] (s22) at (6, 1) {};
		      \node[vertex] (s31) at (3, 0) {};
		      \node[vertex] (s32) at (5, 0) {};
		      \node[vertex, label=90:$t^{3,2}_i$] (s33) at (7, 0) {};
		      \node[vertex] (s41) at (2, -1) {};
		      \node[vertex] (s42) at (4, -1) {};
		      \node[vertex] (s43) at (6, -1) {};
		      \node[vertex, label=90:$t^{4,2}_i$] (s44) at (8, -1) {};

		      \begin{scope}[yshift = -10cm, color=gray]

		      \node[vertex, label=90:$\bar{t}^{1,2}_i$] (sb1) at (5, 2) {};
		      \node[vertex, label=90:$\bar{t}^{2,2}_i$] (sb21) at (4, 1) {};
		      \node[vertex] (sb22) at (6, 1) {};
		      \node[vertex, label=90:$\bar{t}^{3,2}_i$] (sb31) at (3, 0) {};
		      \node[vertex] (sb32) at (5, 0) {};
		      \node[vertex] (sb33) at (7, 0) {};
		      \node[vertex, label=90:$\bar{t}^{4,2}_i$] (sb41) at (2, -1) {};
		      \node[vertex] (sb42) at (4, -1) {};
		      \node[vertex] (sb43) at (6, -1) {};
		      \node[vertex] (sb44) at (8, -1) {};

		      \draw (ui) edge node[pos=0.5, fill=white, inner sep=1pt] {\scriptsize $\tbdd$} (s1);
		      \draw (s1) -- (t) -- (s22) -- (s21);
		      \draw (s21) edge node[pos=0.5, fill=white, inner sep=1pt] {\scriptsize $\tbdd$} (ui);
		      \draw (ui) edge node[pos=0.5, fill=white, inner sep=1pt] {\scriptsize $\tbdd$} (s31);
		      \draw (s31) -- (s32) -- (s33) -- (t) -- (s44) -- (s43) -- (s42) -- (s41);
		      \draw (s41) edge node[pos=0.5, fill=white, inner sep=1pt] {\scriptsize $\tbdd$} (ui);

		      \draw (li) -- (sb1);
		      \draw (sb1) edge node[pos=0.5, fill=white, inner sep=1pt] {\scriptsize $\tbdd$} (t);
		      \draw (t) edge node[pos=0.5, fill=white, inner sep=1pt] {\scriptsize $\tbdd$} (sb22);
		      \draw (sb22) -- (sb21) -- (li) -- (sb31) -- (sb32) -- (sb33);
		      \draw (sb33) edge node[pos=0.5, fill=white, inner sep=1pt] {\scriptsize $\tbdd$} (t);
		      \draw (t) edge node[pos=0.5, fill=white, inner sep=1pt] {\scriptsize $\tbdd$} (sb44);
		      \draw (sb44) -- (sb43) -- (sb42) -- (sb41) -- (li);
		      \end{scope}

		      \draw (s1) edge[red, dashed] node[pos=0.5, label=180:$\bar{c}^{1, 2}_i$] {} (sb31);
		      \draw (s22) edge[red, dashed] node[pos=0.5, label=0:$\bar{c}^{2, 2}_i$] {} (sb21);
		      \draw (s33) edge[red, dashed] node[pos=0.5, label=0:$\bar{c}^{3, 2}_i$] {} (sb1);
		      \end{scope}
		      \end{scope}
		    \end{tikzpicture}

		  \end{center}
		  \caption{An example of a vertex-selection gadget for $\otherN=4$ and $\tbd = 2$.
		  All vertices $x^{j,2}_i$ are drawn in gray.
		  Shortcut edges $c^{j,1}_i$ or $\bar{c}^{j,1}_i$ are drawn in dashed blue, while shortcut edges $c^{j,2}_i$ or $\bar{c}^{j,2}_i$ are drawn in dashed red.
		  To keep the figure readable, only one instead of $\tbd$ paths $S^{j,p}_i$, $\bar{S}^{j,p}_i$, $T^{j,p}_i$, and $\bar{T}^{j,p}_i$ for $j\in [n]$ is drawn in the picture.
		  }
		  \label{fVertexSelectionGadgetFVS}
		\end{figure}
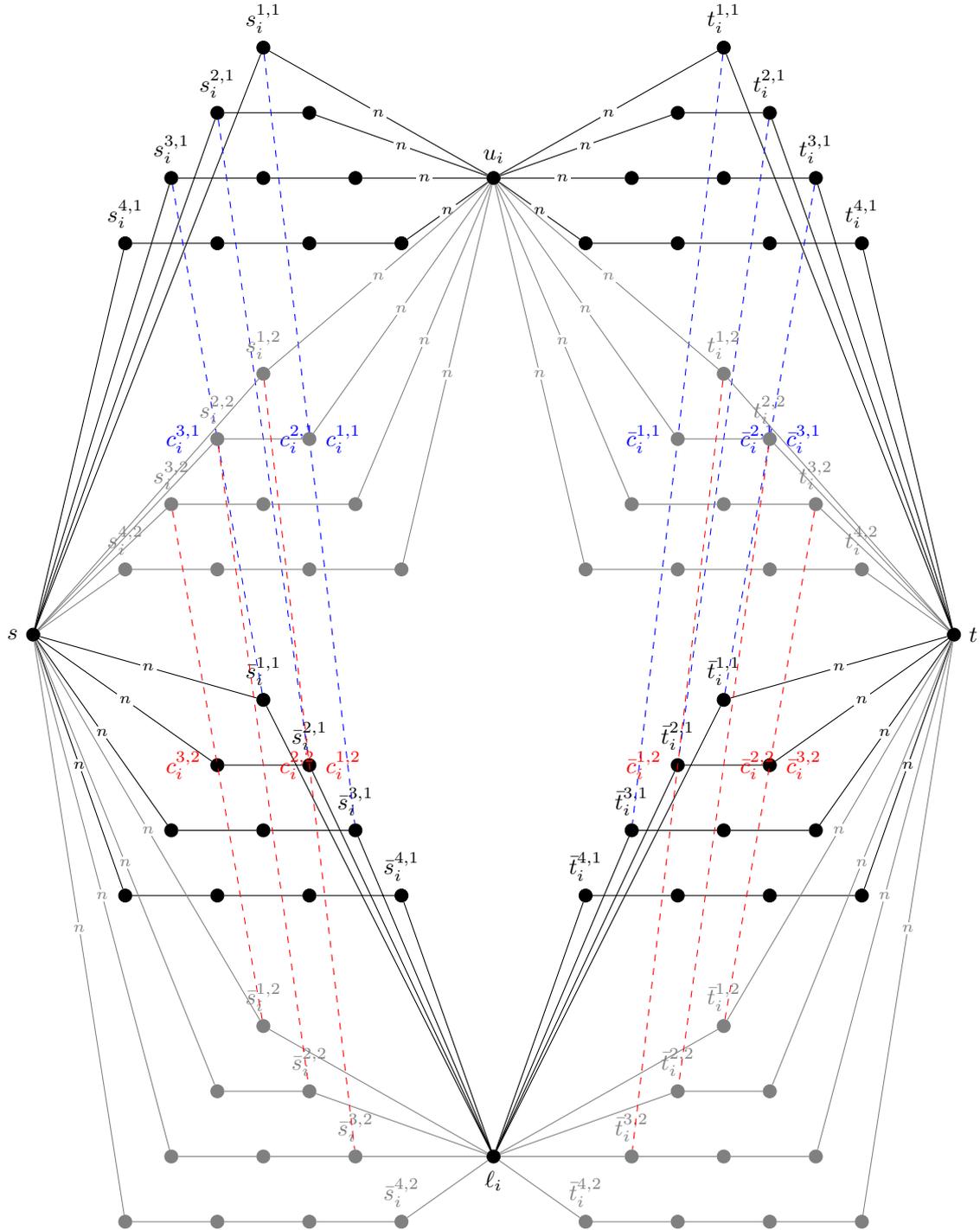

\subsubsection{Edge Gadgets}

For each edge $e = \{v_i, v_j\}$ with $v_i\in V_i$ and $v_j\in V_j$, we add a vertex $v_e$.
This vertex $v_e$ is connected to $t$ by an edge.
We add a path from $u_i$ to $v_e$ of length $\tbdd + \otherN - \idx (v_i)$, a path from $\ell_i$ to $v_e$ of length $\tbdd + \idx (v_i)$, a path from $u_j$ to $v_e$ of length $\tbdd + \otherN - \idx (v_j)$ and a path of from $\ell_j$ to $v_e$ of length $\tbdd + \idx (v_j)$.
We say for any edge $e\in E$, that the edge gadget containing $v_e$ is \emph{corresponding} to $e$.
An example of an edge gadget is shown in Figure~\ref{fedgegadget}.
\begin{figure}
  \begin{center}
		\begin{minipage}{.7\columnwidth}
			\begin{tikzpicture}[xscale=2.5,yscale=1.]
	      \node[vertex, label=180:$v_i$] (ui) at (0, .75) {};
	      \node[vertex, label=180:$\ell_i$] (li) at (0, .25) {};
	      \node[vertex, label=180:$v_j$] (uj) at (0, -.25) {};
	      \node[vertex, label=180:$\ell_j$] (lj) at (0, -.75) {};

	      \node[vertex, label=90:$v_e$] (ve) at (2.5, 0) {};

	      \node[vertex, label=0:$t$] (t) at (3.2, 0) {};

	      \draw (ui) to[out=0,in=135] node[pos=0.5, fill=white, inner sep=1pt] {\scriptsize $\tbdd + \otherN - \idx (v_i)$} (ve);
	      \draw (li) to[out=0,in=160] node[pos=0.5, fill=white, inner sep=1pt] {\scriptsize $\tbdd + \idx (v_i)$} (ve);
	      \draw (uj) to[out=0,in=200] node[pos=0.5, fill=white, inner sep=1pt] {\scriptsize $\tbdd + \otherN - \idx (v_j)$} (ve);
	      \draw (lj) to[out=0,in=225] node[pos=0.5, fill=white, inner sep=1pt] {\scriptsize $\tbdd + \idx (v_j)$} (ve);
	      \draw (ve) -- (t);
	    \end{tikzpicture}
		\end{minipage}
		\begin{minipage}{.29\columnwidth}
			\caption{An example of an edge gadget for the edge $\{v_i, v_j\}$, where $v_i$ belongs to $V_i$ and $v_j$ belongs to $V_j$.}
		  \label{fedgegadget}
		\end{minipage}
  \end{center}
\end{figure}
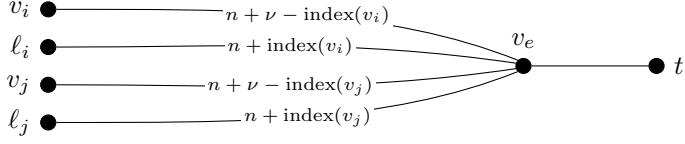

\subsection{Forward Direction}\label{sFVSforward}
  We first show how to construct a $\lambda$-cut $F$ of size $\beta$ from a clique $\{c_1, \dots, c_k\}$ with $c_i \in V_i$ in~$G$.
  For each vertex selection gadget, we add
  \begin{itemize}
    \item the first edge of all $s$-$u_i$-paths of length at most $\tbdd + \idx (c_i) - 1$ (i.e., $\{\{s, s^{j,p}_i\}: j\le \idx (c_i) - 1, 1\le p \le \tbd\}$),
    \item the last edge of all $s$-$\ell_i$-paths of length at most $\tbdd + \otherN - \idx (c_i)$ (i.e., $\{\{ \bar{s}^{j,p}_i, \ell_i\}: j\le \otherN - \idx (c_i), 1\le p \le \tbd\}$),
    \item the last edge of all $u_i$-$t$-paths of length at most $\tbdd + \otherN - \idx (c_i)$ (i.e., $\{\{t^{j,p}_i, t\}: j\le \otherN - \idx (c_i), 1\le p \le \tbd\}$), and
    \item the first edge of all $\ell_i$-$t$-paths of length at most $\tbdd + \idx (c_i) - 1$ (i.e., $\{\{\ell_i, \bar{t}^{j,p}_i\}: j\le \idx (c_i) - 1, 1\le p \le \tbd\}$) to $F$.
  \end{itemize}

  For each edge~$e$ not corresponding to an edge inside the clique $\{c_1, \dots, c_k\}$, we add the edge~$\{v_e, t\}$ to $F$.
  First, we show that $F$ contains exactly $\beta $ edges.

  \begin{lemma}\label{lSizeOfF}
    We have $|F| = \beta = 2k (\otherN -1) \tbd + m - \binom{k}{2}$.
  \end{lemma}

  \begin{proof}
    The cut $F$ contains $2 (\otherN -1) \tbd$ edges from each of the $k$ vertex selection gadgets, and one edge for each of the $m- \binom{k}{2}$ edge gadgets not corresponding to an edge of the form~$\{c_i, c_j\}$.
    Thus, we have $|F| = 2(\otherN -1)\tbd k + m - \binom{k}{2} = \beta$.
  \end{proof}

  It remains to show that any $s$-$t$-path in $H-F$ has length at least $\lambda +1$.
  We do so by first showing that each $s$-$u_i$-path has length at least $\tbdd + \idx (c_i)$, and each $s$-$\ell_i$-path has length at least $\tbdd + \otherN - \idx (c_i)$.
  Afterwards, we also bound the length of an $u_i$-$t$- and an~$\ell_i$-$t$-path from below.
  Before we can show the main results, we first need an auxiliary lemma:

  \begin{lemma}\label{lpathlengthAtLeastn}
    Any $s$-$u_i$-path ($s$-$\ell_i$-, $u_i$-$t$-, or $\ell_i$-$t$-path) in $H-F$ has length at least $\tbdd$.
  \end{lemma}

  \begin{proof}
    We will only prove the lemma for $s$-$u_i$- and $s$-$\ell_i$-paths;
    for the $u_i$-$t$- and $\ell_i$-$t$-paths, the statement can be proven analogously.

    We prove the lemma by contradiction, so let $v\in \{u_i, \ell_i : i \in [k]\}$ be the vertex closest to~$s$ in~$H-F$ among all vertices in $\{u_i, \ell_i : i\in [k]\}$, and assume that the distance from $s$ to~$v$ is~$d<\tbdd$.
    Let $P$ be an $s$-$v$-path of length~$d$.

    By the choice of $v$, no vertex from $\{u_i, \ell_i : i\in [k]\}$ is an interior vertex of $P$.
    Thus, $P$ has to be of the form $s$-$s^{j, p}_i$-$c^{j, p}_i$-$\bar{s}^{\otherN -j, p}_i$-$\ell_i$ (in particular, we have $v= \ell_i$), as all paths $S^{j,p}_i$ and $\bar{S}^{j,p}_i$ are of length at least $\tbdd + 1$.

    However, by the construction of $F$, either the edge $\{s, s^{j,p}_i\}$ (if $j< \idx (c_i)$) or the edge~$\{\bar{s}^{\otherN -j, p}_i, \ell_i\}$ (if $j\ge \idx (c_i)$) is contained in $F$, a contradiction.
  \end{proof}

We can now analyze the length of all $s$-$u_i$- and $s$-$\ell_i$-paths in~$H - F$.

  \begin{lemma}\label{lsupaths}
    Any $s$-$u_i$-path in $H - F$ has length at least $\tbdd + \idx (c_i)$, and any $s$-$\ell_i$-path has length at least $\tbdd + \otherN - \idx (c_i) + 1$.
  \end{lemma}

  \begin{proof}
    We first show that we only need to consider $s$-$u_i$- and $s$-$\ell_i$-paths containing no vertex from $\{u_j, \ell_j: j\in [k]\}$ as an interior vertex.
    To see this, first note that the only connections between a vertex from $\{u_i, \ell_i\}$ to a vertex from $\{u_j, \ell_j: j\in [k]\setminus \{i\}\}$ with $i\neq j$ in $G-\{s\}$ are through an edge gadget and thus of length at least~$2\tbdd$ or through $t$ and thus by Lemma~\ref{lpathlengthAtLeastn} of length at least~$2\tbdd$.
    Also any $u_i$-$\ell_i$-path is of length at least $\tbdd$ (as any path leaving $u_i$ starts with a path of at least $\tbdd - 1$ vertices of degree two).
    Thus, any $s$-$\ell_i$- or $s$-$u_i$-path containing~$u_i$ or $\ell_i$ as an interior vertex is of length at least $2\tbdd > \tbdd + \max\{\idx (c_i), \otherN - \idx (c_i) + 1\}$ by Lemma~\ref{lpathlengthAtLeastn}.

    We first consider paths not containing shortcut edges.
    Any $s$-$u_i$-path not containing a shortcut edge is of the form $s$-$S^{j,p}_i$-$u_i$.
    By the construction of~$F$, such a path has length at least~$\tbdd + \idx (c_i)$.
    Any shortest $s$-$\ell_i$-path not containing a shortcut edge is of the form $s$-$\bar{S}^{j,p}_i$-$\ell_i$.
    By the construction of $F$, such a path has length at least $\tbdd + \otherN - \idx (c_i) + 1$.
    Now, consider an~$s$-$v$-path $P$ with~${v\in \{u_i , \ell_i\}}$ containing a shortcut edge $c^{j,p}_i$ for some $j\in [\otherN -1]$ and $p\in [\tbd]$.
    By the construction of $F$, we have that either $\{s, s^{j,p}_i\}$ (if~$j \le \idx (c_i) - 1$) or $\{\bar{s}^{\otherN-j,p}_i, \ell_i\}$ (if~$j \ge \idx (c_i)$) is contained in $F$.
    Thus, $F$ has to be of the form $s$-$\bar{S}^{\otherN-j, p}_i$-$c^{j, p}_i$-$S^{j, p}_i$-$u_i$, which is of length $(\tbdd + \otherN-j -1) + 1 + (\tbdd + j -1) = 2\tbdd + \otherN -1 > \tbdd + \idx (c_i)$.
  \end{proof}

Analogously, we can also show the following bound for all $u_i$-$t$- and $\ell_i$-$t$-paths in~$H - F$.

\begin{lemma}\label{lutpaths}
  Any $u_i$-$t$-path in $H-F$ has length at least $\tbdd + \otherN - \idx (c_i) + 1$, and each~$\ell_i$-$t$-path has length at least $\tbdd + \idx (c_i)$.
\end{lemma}

The correctness of the forward direction now easily follows.

\begin{lemma}\label{lforwardfvs}
  If $G$ contains a clique of size $k$, then $H$ contains a $\lambda$-cut of size~$\beta$.
\end{lemma}

\begin{proof}
  By Lemma~\ref{lSizeOfF}, we have $|F| = \beta$, so it suffices to show that $F$ is a $\lambda$-cut.
  Consider any~\mbox{$s$-$t$-path} $P$ in $H-F$.
  Then $P$ passes through a vertex $u_i$ or $\ell_i$ for some $i\in [k]$, as any~\mbox{$s$-$t$-path} in $H$ passes through a vertex of the form $u_j$ or $\ell_j$.
  If $P$ passes through $u_i$, then we get by Lemma~\ref{lsupaths} and~\ref{lutpaths} that the length of $P$ is at least~$\tbdd + \idx (c_i) + \tbdd + \otherN -\idx (c_i) + 1 = \lambda +1$.
  If $P$ passes through $\ell_i$, then we get by Lemma~\ref{lsupaths} and~\ref{lutpaths} that the length of $P$ is at least~$\tbdd + \otherN -\idx (c_i) + 1 + \tbdd + \idx (c_i) = \lambda +1$.
\end{proof}

\subsection{Backward Direction}\label{sFVSbackward}

We now turn to the backward direction, i.e., that any $\lambda$-cut $F$ of size at most $\beta$ in $H$ implies a clique of size $k$ in $G$.
In order to do so, we first show that $F$ has to have a certain structure.
We begin with the structure of $F$ in vertex-selection gadgets.

\begin{lemma}\label{ltriv}
  Let $F$ be a $\lambda$-cut in $H$.
  For any vertex-selection gadget $S_i$ and any $j\in [\otherN - 1]$, the cut $F$ contains an edge of every $s$-$u_i$-path of length at most $\tbdd + j$, or an edge of every~$u_i$-$t$-path of length at most $\otherN - j +\tbdd$.
  An analogous statement holds for $\ell_i$.
\end{lemma}

\begin{proof}
  If this is not the case, then the $s$-$u_i$-path of length at most $\tbdd + j$ and the $u_i$-$t$-path of length at most $\otherN - j + \tbdd$ yield an $s$-$t$-path of length at most $\tbdd+ j + \otherN - j + \tbdd = \lambda$.
\end{proof}

This yields a bound on the number of edges that $F$ contains in any vertex-selection~gadget.

\begin{corollary}\label{cvsg}
  Let $F$ be a $\lambda$-cut in $H$.

  Then $F$ contains at least $2 (\otherN - 1)\tbd$ edges inside each vertex-selection gadget.
\end{corollary}

\begin{proof}
  By applying Lemma \ref{ltriv} for each $j\in [\otherN - 1]$, we get that $F$ contains an edge of every~$s$-$u_i$-path of length $\tbdd + j$ or an edge of every $u_i$-$t$-path of length $\otherN - j +\tbdd$.
  This sums up to $2(\otherN-1) \tbd$, as there are $\tbd$ copies of each path.
\end{proof}

Applying the lower bound from Corollary~\ref{cvsg} for each vertex-selection and using $|F|\le \beta $ yields also an upper bound for the number of edges in any vertex-selection gadget via a simple counting argument.

\begin{corollary}\label{cNumEdgesInVSG}
  Let $F$ be a $\lambda$-cut in $H$ of size at most $\beta$.
  Then $F$ contains less than $(2\otherN -1)\tbd$ edges inside a vertex-selection gadget.
\end{corollary}

\begin{proof}
  By Corollary \ref{cvsg}, any vertex-selection gadget contains at least $2(\otherN -1)\tbd $ edges.
  Thus, for any fixed vertex-selection gadget, there are at least $\gamma := (k-1) 2(\otherN-1) \tbd$ edges contained in other vertex-selection gadgets.
  Therefore, at most $\beta - \gamma = 2k (\otherN-1) \tbd + m -\binom{k}{2} - (k-1) 2 (\otherN - 1) \tbd = 2 (\otherN-1)\tbd + m  -\binom{k}{2} < (2\otherN-1)\tbd$ edges are contained in a single vertex-selection gadget.
\end{proof}

We are now ready to show that the set of edges that $F$ contains inside each vertex-selection gadget is similar to the set of edges the $\lambda$-cut constructed from a clique in Section \ref{sFVSforward} contains in a vertex-selection gadget.

\begin{corollary}\label{csel}
  Let $F$ be a $\lambda$-cut in $H$ of size at most $\beta$.
  For each vertex-selection gadget $S_i$, there exists a number $c_i$ such that $F$ contains one edge from each $s$-$u_i$- ($\ell_i$-$t$-)path of length at most~$\tbdd + c_i - 1$, and one edge from each $s$-$\ell_i$- ($u_i$-$t$)-path of length at most $\otherN - c_i + \tbdd$.
\end{corollary}

\begin{proof}
  By Lemma \ref{ltriv}, there exists some $c_i$ ($d_i$) such that $F$ contains an edge from every~$s$-$u_i$- ($\ell_i$-$t$-)path of length at most $\tbdd + \idx (c_i) - 1$ ($\tbdd + \otherN - \idx(d_i)$), and one edge from each $s$-$\ell_i$- ($u_i$-$t$)-path of length at most $\otherN - \idx (c_i) + \tbdd$ ($\idx (d_i) + \tbdd - 1$).

  It remains to show that $c_i = d_i$.
  By Lemma~\ref{ltriv} and Corollary~\ref{cNumEdgesInVSG}, we know that in each vertex selection gadget, $F$ contains less than $(2\otherN -1)\tbd $ edges, and $2(\otherN -1)\tbd$ of these lie inside the paths~$S^{j,p}_i$, $\bar{S}^{j, p}_i$, $T^{j,p}_i$, and $\bar{T}^{j, p}_i$.
  Thus, by the pigeonhole principle, for each~$j\in [\otherN -1]$, there exists a~$p\in [\tbd]$ such that the shortcut edges $ c^{j, p}_i$ and $\bar{c}^{j, p}_i$ are not contained in $F$.
  If~$\idx (c_i) < \idx (d_i)$, then $s$-$s^{\idx(c_i)p}_i$-$c^{\idx (c_i)p}_i$-$\ell_i$-$t$ is an $s$-$t$-path of length at most $3+ \tbdd + \otherN < \lambda$.
  If~$\idx (c_i) > \idx (d_i)$, then $s$-$\ell_i$-$\bar{c}^{\otherN - \idx (c_i), p}_i$-$t$ is an $s$-$t$-path of length at most $\tbdd + \otherN + 3< \lambda$.
\end{proof}

If~$F$ contains an edge from each~\mbox{$s$-$u_i$-path} of length at most $\tbdd + c_i - 1$ and an edge from each~\mbox{$s$-$\ell_i$-path} of length at most $\otherN - c_i + \tbdd$ in a vertex selection gadget~$S_i$, then we say that~$S_i$ selects the vertex $\vertex_i (c_i)\in V_i$.
We now turn our attention to the edge gadgets, and show that~$F$ contains one edge from each edge gadget not corresponding to an edge between two selected vertices.

\begin{lemma}\label{ledge}
  Let $i,j\in [k]$ with $i \neq j$.
  Then any $\lambda$-cut $F$ of size at most $\beta$ in $H$ contains one edge from each edge gadget not corresponding to an edge $\{c_i, c_j\}$, where $c_i$ and $c_j$ are the vertices selected by the vertex-selection gadgets $S_i$ and $S_j$.
\end{lemma}

\begin{proof}
  Let $e=\{v_i, v_j\}$ be an edge with $e\neq \{c_i,c_j\}$, where $v_i\in V_i$ and $v_j\in V_j$.
  Due to symmetry, we may assume without loss of generality that $v_i\neq c_i$.
  If $\idx (v_i) < \idx(c_i)$, then $s$-$\ell_i$-$v_e$-$t$ is an $s$-$t$-path of length $\tbdd + \otherN - \idx (c_i) +1 + x + \tbdd < 2\tbdd + \otherN + 1= \lambda + 1$.
  If $\idx(v_i)>\idx (c_i)$, then $s$-$u_i$-$v_e$-$t$ is an $s$-$t$-path of length $\tbdd + \idx (c_i) + \otherN - x + \tbdd + 1= 2\tbdd + \otherN + \idx (c_i) - \idx (v_i) + 1<\lambda + 1$.
\end{proof}

The correctness of the backward direction is now easy to show.

\begin{lemma}\label{lbackwardfvs}
  If $H$ contains a $\lambda$-cut $F$ of size $\beta$, then $G$ contains a clique of size $k$.
\end{lemma}

\begin{proof}
  By Corollary \ref{csel}, each vertex-selection gadget selects a vertex.
  By Lemma~\ref{ledge}, each edge from $G$ not between two selected vertices induces an edge inside $F$.
  Thus, $F$ can be of size $\beta $ if and only there are $\binom{k}{2}$ edges between selected vertices in $G$.
  In other words, the selected vertices form a clique.
\end{proof}

\subsection{Feedback Vertex Number}

It remains to analyze the time required to compute the reduction and to show that the feedback vertex number of $H$ is bounded in terms of $k$.
We start with the running time.

\begin{observation}
\label{obs:fvnrunningtime}
The given reduction of \textsc{Multicolored Clique} parameterized by solution size~$k$ to \lbc{} parameterized by feedback vertex number can be computed in~$O(k \cdot m \cdot n)$ time.
\end{observation}

\begin{proof}
Let~$(G=(V,E),k)$ be the input instance of \textsc{Multicolored Clique} and let~$n = |V|$ and~$m = |E|$.
Observe that the resulting graph~$H$ contains~$k$ vertex-selection gadgets and~$m$ edge gadgets.
Note further that the resulting instance~$(H,s,t,\beta,\lambda)$ can be computed in time linear in the size of~$H$.
Lastly, notice that each vertex-selection gadget contains~$O(m \cdot n)$ vertices and edges and each edge gadget contains~$O(n)$ vertices and edges.
Thus,~$H$ contains~$O(k\cdot n \cdot m)$ vertices and edges and can be computed in~$O(k\cdot n \cdot m)$ time.
\end{proof}

Last but not least, we need to analyze the feedback vertex number of~$H$.
We do this by simply giving a feedback vertex set of size $O(k)$.

\begin{lemma}\label{lfvs}
  The set $X:=\{s, t\}\cup \{ u_i, \ell_i : i\in [k]\}$ is a feedback vertex set in~$H$.
\end{lemma}

\begin{proof}
  Note that all vertices from $V_H\setminus X$ are contained in a path $S^{i,j}_p$, $\bar{S}^{i,j}_p$, $T^{i,j}_p$, or $\bar{T}^{i,j}_p$ or contained in an edge gadget.
  All edges from the graph $H-X$ not contained in one of these paths or an edge gadget are the shortcut edges $c^{i,j}_p$ and $\bar{c}^{i,j}_p$.

  Thus, there are only three kinds of different connected components in $H - X$, and all of them are trees:

  \begin{itemize}
    \item Clearly, edge gadgets are trees.
    \item Components of the form $\{S^{j, p}_i, c^{j, p}_i, \bar{S}^{\otherN -j, p}\}$ or $\{T^{j, p}_i, \bar{c}^{j, p}_i, \bar{T}^{\otherN -j, p}\}$ with $1\le j \le \otherN -1 $ are paths.
    \item Components of the form $S^{\otherN, p}_i$, $T^{\otherN, p}_i$, $\bar{S}^{\otherN, p}_i$, and $\bar{T}^{\otherN, p}$ are paths.\qedhere
  \end{itemize}
\end{proof}

Combining Lemmata~\ref{lforwardfvs}, \ref{lbackwardfvs}, and \ref{lfvs} with Observation \ref{obs:fvnrunningtime} yields our desired main result.

\begin{theorem}\label{thm:LCutWhardFVS}
  \textsc{Length-Bounded Cut} parameterized by feedback vertex number $k$ is W[1]-hard.
  Assuming ETH, it cannot be solved in $f(k) \cdot n^{o(k)}$ time for any computable function~$f$.
\end{theorem}

\section{Polynomial-Time Algorithm on Proper Interval Graphs}
In this section we will present a polynomial-time algorithm for \lbc{} on proper interval graphs.
The algorithm is a dynamic program that stores for each vertex~$v$ and each possible distance~$d$ ($2 \leq d \leq \lambda$) the minimal size of a cut that makes each vertex in a particular subset of vertices including~$v$ have distance at least~$d$ from~$s$.

Observe that we can assume without loss of generality that~$\s{s} \leq \s{t}$ as we can otherwise ``mirror'' the graph by setting~$\s{v} = -\e{v}$ and~$\e{v} = -\s{v}$ for each vertex~$v \in V$.
It is folklore that one can assume that all~\s{}-values are distinct, that is,~$|\{\s{v} \mid v \in V\}| = |V|$.
We now sort all the vertices in $V\setminus \{s, t\}$ by their respective~\s{}-value in increasing order and rename the vertices such that~$v_i$ is the~$i\textsuperscript{th}$ vertex in this order.
Thus, we have $V = \{s, t\} \cup \{v_1, \dots, v_{n-2}\}$, and $\s{v_i} < \s{v_{i+1}}$ for all $i\in [n-3]$.
We will first show that we can safely ignore all vertices~$v$ with~$\e{v} < \s{s}$ or~$\e{t} < \s{v}$.

\begin{lemma}
\label{lem:borders}
Let~$I = (G=(V,E), s,t, \beta, \lambda)$ be an instance of \lbc{} where~$G$ is an interval graph and~$\s{s} < \s{t}$ in the interval representation.
Let~$L = \{u\in V \mid \e{u} < \s{s}\}$ and~$R = \{u \in V \mid \e{t} < \s{u}\}$.
Then~$I' = (G[V \setminus (L \cup R)], s, t, \beta, \lambda)$ is an equivalent instance of \lbc.
\end{lemma}

\begin{proof}
Let~$I,I',G,s,t,\beta,\lambda,L,$ and~$R$ be as defined above.
We will first show that the instance $I_L = (G[V \setminus R], s, t, \beta, \lambda)$ is an equivalent instance.
The argumentation for then removing~$L$ from~$I_L$ to obtain the equivalent instance~$I'$ is analogous and hence skipped here.
First observe that~${s,t \notin L \cup R}$ and hence~$I_L$ and~$I'$ are instances of \lbc.
Observe that deleting vertices from any input graph cannot decrease the distance between any pair of vertices and hence if~$I$ is a yes-instance, then so is~$I_L$.
Hence it remains to show that if~$I_L$ is a yes-instance, then so is~$I$.
Assume towards a contradiction that this is not the case and hence~$I_L$ is a yes--instance and~$I$ is a no-instance.
Then there is a set~$F_{I_L}$ of~$\beta$ edges in~$G[V \setminus R]$ such that the distance between~$s$ and~$t$ in~$G_L = (V \setminus R, E \setminus (F_{I_L} \cup \{\{u,v\} \in E \mid u \in R\}))$ is at least~$\lambda + 1$.
Since~$I$ is a no-instance, there is a path~$P$ of length at most~$\lambda$ between~$s$ and~$t$ in~$G^* = (V, E \setminus F_{I_L})$.
As~$G_L$ and~$G^*$ only differ in~$R$, each path of length at most~$\lambda$ between~$s$ and~$t$ in~$G^*$ contains at least one vertex from~$R$.
We will show that~$\deg_G(t) \leq |F_{I_L}|$ and hence there is an~$s$-$t$-cut of size at most~$\beta$ in~$G$ and thus~$I$ is a yes-instance.
This contradicts the assumption that~$I$ is a no-instance and hence finishes the proof that~$I_L$ is equivalent to~$I$.

		We start by given some basic notation for the proof to come.
We use sets of vertices that have a certain distance from~$s$ in some subgraph~$H$ of~$G$.
To this end, we define~$X_H^p = \{u \in V \mid \dist_H(s,u) = p\}$ for each distance~$p$.
Analogously, we define~$X_H^{\leq p} = \{u \in V \mid \dist_H(s,u) \leq p\}$ and~$X_H^{\geq p} = \{u \in V \mid \dist_H(s,u) \geq p\}$.

		Let~$d = \dist_{G^*}(s,t)$ and let~$t'$ be the vertex in~$P$ with maximum~$\s{t'}$.
Since~$P$ contains a vertex from~$R$, it holds that~$\s{t'} > \e{t}$ and hence~$t' \notin N_G(t)$.
Since~$t'$ is on a shortest~$s$-$t$-path in~$G^*$ and~$t' \notin N_G(t)$ it holds that~$t' \in X_{G^*}^{\leq d-2}$.
Now consider the set~$K$ of vertices that are part of a shortest~$s$-$t'$-path in~$G^*$ and that are neighbors of~$t$ in~$G$.
By construction~$K \subseteq X_{G^*}^{\leq d-3}$ and for each~$y \in [\s{t},\e{t}]$ there is a vertex~$v \in K$ with~$y \in [\s{v},\e{v}]$.
We will next show that~$|F_L| \geq \deg_G(t)$.
To this end, observe first that for each vertex~$u \in N_G(t) \cap X_{G^*}^{\leq d-2} \supseteq K$ it holds that~$\{u,t\} \in F_L$.
Next observe that for each~$u \in N_G(t)$ it holds by definition that~$[\s{u},\e{u}] \cap [\s{t},\e{t}] \neq \emptyset$ and hence for each~$u \in N_G(t) \cap X_{G^*}^{\geq d-1}$ there is a vertex~$v \in K$ with~$\{u,v\} \in E$.
Since~$\dist_{G^*}(s,u) \geq d-1 > d-3 + 1 \geq \dist_{G^*}(s,v) + 1$ it holds that~$\{u,v\} \in F_L$.
It is then easy to verify that~$\beta = |F_L| \geq \deg_G(t)$ and hence there is a trivial~$s$-$t$-cut of size~$\beta$ in~$G$ that just removes all incident edges of~$t$.
This contradicts the assumption that~$I$ is a no-instance and thus concludes the proof.
\end{proof}

Using \Cref{lem:borders}, we will always assume that there is no vertex $v$ with $\e{v} < \s{s}$ or $\s{v} > \e{t}$.
We next show that there always exists a solution in which the distance from~$s$ to~$v_j$ is non-decreasing in~$j$.

\begin{lemma}
\label{lem:monotone}
Let~$G=(V,E)$ be a proper interval graph such that there is no vertex~$v$ with $ \e{t} < \s{v}$ or $\s{s} > \e{v}$ and let~$F$ be a set of edges such that in~$G' = (V, E \setminus F)$ the vertex~$t$ has at least some distance~$d$ from~$s$.
There is a set~$F'$ of edges with~$|F'| \le |F|$ such that for~$G'' = (V,E\setminus F')$ it holds that~$\dist_{G''}(s, t) \geq d$ and for each~$v_i,v_j \in V\setminus\{s,t\}$ with~$\s{v_i} < \s{v_j}$ it holds that~$\dist_{G''}(s,v_i) \leq \dist_{G''}(s,v_j)$.
\end{lemma}

\begin{proof}
Let~$G,s,t,F,G',$ and~$d$ be as defined above.
For each vertex~${v\in V}$ of a graph~$H = (V, E_H)$, we define a specific distance~$\D_H(v)$.
We define~$\D_H(v)$ to be the length of a shortest path~$P=(s=u_0,u_1,u_2,\ldots,u_\alpha=v)$ from~$s$ to~$v$ in a graph~$H$ on the same set of vertices as~$G$ such that for all~$\gamma \in [\alpha-1]$ it holds that~$\s{u_{\gamma}} < \s{u_\gamma+1}$.
As a special case, if~$u_\alpha = t$, then we only require that for all~$\gamma \in [\alpha-2]$ it holds that~$\s{u_{\gamma}} < \s{u_\gamma+1}$.
If no such path exists, then we define~$\D_H (v) \coloneqq \infty$.
Observe that for each graph~$H$ it holds that~$\D_H(s) = 0$ and~$\D_{H}(v) \geq \dist_{H}(s,v)$.
Let~$\mathcal{G} = \{G^* = (V, E^*) \mid E^* \subseteq E \land |E^*| \ge |E \setminus F|\}$.
We will present a sequence of graphs~$(G'=G_1,G_2,\ldots G_k)$ such that
\begin{enumerate}
	\item $G_\ell = (V, E_\ell) \in \mathcal{G}$ for each~$\ell \in [k]$, \label{enum:G}
	\item $\D_{G_{\ell}}(t) \leq \D_{G_{\ell+1}}(t)$ for each~$\ell \in [k-1]$, and \label{enum:inductionStep}
	\item $\D_{G_k}(v) \leq \D_{G_k}(w)$ for all~$v,w \in V \setminus \{s,t\}$ with~$b_v < b_w$.\label{enum:monotone}
\end{enumerate}

\begin{claim}
If such a sequence exists, then $G'':= G_k$ satisfies the lemma.
\end{claim}

\begin{proof}
   \renewcommand{\qedsymbol}{(of Claim 1)~$\diamond$}
  First, we show that~$\dist_{G_k}(s,v) = \D_{G_k}(v)$ for all~$v$.
  Assume towards a contradiction that there is some vertex~$v \in V \setminus \{t\}$ with~$\dist_{G_k}(s,v) \neq \D_{G_k}(v)$, and $b_v $ minimal among those.
  Consider any shortest~$s$-$v$-path~$P$ in~$G_k$.
  Let $w$ be the first vertex on $P$ with $\dist_{G_k} (s, w) \neq \D_{G_k}(w)$, and let $w'$ be its predecessor in $P$.
  By the definition of $w$, we have $\dist_{G_k} (s, w') = \D_{G_k} (w')$ and $b_{w} < b_{w'}$.
  It follows that $\D_{G_k}(w) > \dist_{G_k}(s,w) = \dist_{G_k}(s,w') + 1 = \D_{G_k}(w') + 1$, a contradiction to \eqref{enum:monotone} and $\s{w} < \s{w'}$.
  Now consider~$t$ and any shortest~$s$-$t$-path~$P$ in~$G_k$.
  Let $v$ be the last inner vertex~$v$ in~$P$ (the predecessor of~$t$).
  We have shown that~$\dist_{G_k}(v) = \D_{G_k}(v)$ and hence~$\dist_{G_k}(t) = \dist_{G_k}(v) + 1 = \D_{G_k}(v) + 1 = \D_{G_k}(t)$.
  The last step follows from the fact that~$\D_{G_k}(t) \leq \D_{G_k}(v) +1$ as~$v$ is a neighbor of~$t$ in~$G_k$ and the special case in the definition of~$\D$ that allows to ignore~$\s{t}$.

  The claim now easily follows.
  Using \eqref{enum:inductionStep}, we get~$\dist_{G_k}(s,t) = \D_{G_k}(t) \geq \D_{G_{k-1}}(t) \geq \ldots \geq \D_{G_1}(t) = \D_{G'}(t) \geq \dist_{G'}(s,t) \geq d$, and from \eqref{enum:monotone} it follows for all~$v,w \in V \setminus \{s,t\}$ that~$\dist_{G''}(s, v) = \D_{G''}(v) \leq \D_{G''}(w) = \dist_{G''}(s,w)$ if~$\s{v} < \s{w}$.
\end{proof}

We will now describe how to obtain the sequence of graphs~$(G'=G_1,G_2,\ldots, G_k)$.
The main idea is to apply a number of local changes such that conflicts of \eqref{enum:monotone} are eliminated until there are non left.
To this end we need a rather technical order over the graphs in~$\mathcal{G}$.
We say that~$(V,E_\alpha) = G_\alpha <_{\tor} G_\gamma = (V,E_\gamma)$ for~$G_\alpha,G_\gamma \in \mathcal{G}$ if and only if
\begin{itemize}
  \item $|E_\alpha| > |E_\gamma|$,
  \item $|E_\alpha| = |E_\gamma|$, and there exists a~$v\in V \setminus \{t\}$ such that~$\D_{G_\alpha}(v) < \D_{G_\gamma}(v)$ and~$\D_{G_\alpha}(w) = \D_{G_\gamma}(w)$ for all~$w\in V \setminus \{t\}$ with $\s{w} < \s{v}$, or
  \item $|E_\alpha| = |E_\gamma|$,~$\D_{G_\alpha}(v) = \D_{G_\gamma}(v)$ for all~$v \in V \setminus \{t\}$ and~$\D_{G_\alpha}(t) < \D_{G_\gamma}(t)$.
\end{itemize}
Notice that~$<_{\tor}$ defines a total preorder on~$\mathcal{G}$.

Let~$G_\ell$ be a graph in the sequence.
If~$G_{\ell}$ satisfies \eqref{enum:monotone}, then we have found the last graph in the sequence and thus are done.
Otherwise, we will describe how to obtain another graph~$G_{\ell+1} \in \mathcal{G}$ such that \eqref{enum:inductionStep} holds for~$G_\ell$ and~$G_{\ell+1}$ and~$G_{\ell+1} <_{\tor} G_{\ell}$.
Since~$<_{\tor}$ is a total preorder, we can only build a finite sequence and hence at some point a graph has to satisfy~\eqref{enum:monotone} as otherwise we could continue the sequence infinitely.
Since~$G_\ell$ does not satisfy~\eqref{enum:monotone}, there is some~$j$ such that~$\D_{G_\ell}(v_j) > \D_{G_\ell}(v_{j+1})$.

Let $F_\ell \coloneqq E \setminus E_\ell$ be the set of edges such that $G_\ell = G - F_\ell$.
Since $G_\ell \in \mathcal{G}$, we know that~$|F| \ge |F_\ell|$.
We define
\[
  X:= \{x\in N_G (v_{j+ 1}) \mid (\s{x} < \s{v_j} \lor x = s) \land  \{v_j, x\} \in F_\ell \land \{v_{j+1}, x\}\in E\setminus F_\ell\}
\]
and
\[
  Y:= \{y\in N_G (v_j) \mid (\s{y} > \s{v_{j+1}} \lor y = t) \land \{v_{j+1}, y\} \in F_\ell \land \{v_j, y\} \in E \setminus F_\ell\}\,.
\]

		See \Cref{fig:XY} for an example.
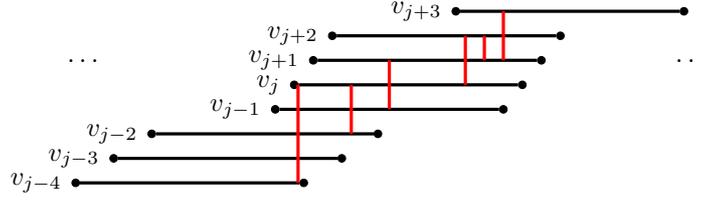
\begin{figure}
  \begin{center}
    \begin{tikzpicture}[yscale = 1.3]

		\node at (-2.5,.5) {$\dots$};
\tikzDrawIntervall{-.75}{-2.625}{0.375}{}{v_{j-4}}\tikzDrawIntervall{-.5}{-2.125}{0.875}{}{v_{j-3}}
\tikzDrawIntervall{-.25}{-1.625}{1.35}{}{v_{j-2}}
\tikzDrawIntervall{0}{0}{3}{}{v_{j-1}}
\tikzDrawIntervall{0.25}{.25}{3.25}{}{v_j}
\tikzDrawIntervall{0.5}{.5}{3.5}{}{v_{j+1}}
\tikzDrawIntervall{0.75}{.75}{3.75}{}{v_{j+2}}
\tikzDrawIntervall{1}{2.375}{5.375}{}{v_{j+3}}
\node at (5.5,.5) {$\dots$};

		\tikzstyle{DeletedEdge}= [red, very thick];

		\draw (0.3, 0.25) edge[DeletedEdge] (0.3, -.75);
\draw (1.5,0) edge[DeletedEdge] (1.5, .5);
\draw (3, 0.5) edge[DeletedEdge] (3, 1);

		\draw (1, 0.25) edge[DeletedEdge] (1, -.25);

		\draw (2.5, 0.25) edge[DeletedEdge] (2.5, 0.75);
\draw (2.75, 0.5) edge[DeletedEdge] (2.75, 0.75);
    \end{tikzpicture}

		  \end{center}
  \caption{An example for $X$ and $Y$.
  Red edges are contained in $F_\ell$.
  Edges in $E\setminus F_\ell$ are for the sake of readability not drawn.
  Observe that~$X= \{v_{j-2}\}$ and $Y = \{v_{j+3}\}$.}
  \label{fig:XY}
\end{figure}
We distinguish two cases:

\noindent{\bfseries Case 1 ($|X| \ge |Y|$):}
		We set $F_{\ell+1} \coloneqq (F_\ell \setminus \{\{v_j, x\} \mid x\in X\}) \cup \{\{v_j, y\} \mid y\in Y\}$, and $G_{\ell + 1} \coloneqq G-F_{\ell +1}$.
Since~$|X| \ge |Y|$, we have $|F_{\ell+1} | \le |F_\ell|$ and hence~$G_{\ell+1} \in \mathcal{G}$.
Clearly, for all~$v\in V\setminus \{t\}$ with~$\s{v} < \s{v_j}$, we have $\D_{G_{\ell + 1}} (v) = \D_{G_\ell} (v)$, as $F_{\ell + 1}$ and~$F_\ell$ only differ in edges incident to~$v_j$.
Let $w$ be the predecessor of $v_{j+1}$ in a shortest monotone $s$-$v_{j+1}$-path.
This vertex~$w$ is adjacent to $v_j$ in $G$ as there is no vertex $v\in V$ with $\e{v} < \s{s}$, and therefore~$w$ is contained in~$X$ as we have $\D_{G_\ell} (v_j) > \D_{G_\ell} (v_{j+ 1})$.
Therefore, we have $\D_{G_{\ell + 1}} (v_j) = \D_{G_\ell} (v_{j+1}) < \D_{G_{\ell}} (v_j)$, and thus we have~$G_{\ell + 1} <_\tor G_{\ell}$.

		It remains to show that $\D_{G_{\ell + 1} } (t) \ge \D_{G_\ell} (t)$.
Consider a shortest monotone $s$-$t$-path~$P$ in~$G_{\ell + 1}$.
If $P$ does not pass through $v_j$, then it is also a monotone $s$-$t$-path in $G_\ell$.
Otherwise, note that~$\D_{G_{\ell +1} } (v_j) = \D_{G_\ell} (v_{j+1})$.
Let $z$ be the successor of $v_j$ in $P$.
Note that $\{v_{j+1}, z\}\in E$ as there is no vertex $v\in V$ with $\s{v} > \e{t}$.
Then $z\notin Y$, as $\{v_j, y\} \in F_{\ell +1}$ for all~$y \in Y$.
Thus, we have~$z = v_{j+ 1}$ or $\{v_{j+1}, z\} \in E_\ell$.
Thus, we have $\D_{G_\ell} (z) \le \D_{G_\ell} (v_{j+1}) + 1 = \D_{G_{\ell + 1}} (v_j) + 1$, and hence there is a monotone path $P'$ in $G_\ell$ that passes through~$z$ and then continues as $P$ after $z$ which is not longer than $P$.

\noindent{\bfseries Case 2 ($|X| < |Y|$):}
		We set $F_{\ell + 1} \coloneqq (F_\ell \setminus \{\{v_{j+1}, y\} \mid y\in Y\}) \cup \{\{v_{j+1}, x\} \mid x\in X\}$.
Since~${|X| < |Y|}$, we have that $|F_{\ell + 1}| < |F_\ell|$ and therefore~$G_{\ell +1} <_\tor G_\ell$.
It remains to show that~$\D_{G_{\ell +1 }} (t) \ge \D_{G_\ell} (t)$.
Let $P$ be a shortest monotone $s$-$t$-path in $G_{\ell + 1}$.
If $P$ does not pass through $v_{j+1}$, then $P$ is also a monotone $s$-$t$-path in $G_{\ell}$, and so $\D_{G_{\ell +1 }} (t) \ge \D_{G_\ell} (t)$.
Otherwise, let~$w$ be the predecessor of~$v_{j+1}$ in $P$.
Since there is no vertex $v\in V$ with $\e{v} < \s{s}$, we have~${\{v_j, w\} \in E}$.
Thus, we get by the definition of $X$ that $\{v_j, w\}\in E_\ell$ and thus, $\D_{G_{\ell + 1}} (v_{j+1}) \ge \D_{G_\ell} (v_{j})$.
Let~$z$ be the successor of~$v_{j+ 1}$ in $P$.
If $z\notin Y$, then $\{v_{j+1}, z\} \in E_\ell$, and therefore, we get a shorter monotone path in $G_\ell$ by replacing the $s$-$v_{j+1}$-path in $P$ by a shortest monotone~$s$-$v_{j+1}$-path in $G_\ell$.
Otherwise, we have~$\{v_j, z\} \in E_\ell$ and thus~$\D_{G_\ell} (z) \le \D_{G_\ell} (v_{j}) + 1 \le \D_{G_{\ell + 1}} (v_{j+1}) + 1 = \D_{G_{\ell + 1}} (z)$.
The last equality follows from the assumption that $P$ is a shortest monotone~$s$-$t$-path in~$G_{\ell+1}$.
Hence there is a monotone path $P'$ that passes through $z$ and then continues as $P$ after~$z$ which is not longer than~$P$.
\end{proof}

We are now in a position to state the main theorem of this section.

\begin{theorem}
\label{thm:lcutPolyPINT}
\lbc{} can be solved in~$O(n^3 \cdot m)$ time if the input graph is a proper interval graph.
\end{theorem}

\begin{proof}
We assume that there is no~$s$-$t$-cut of size at most~$\beta$ in the input graph~$G$ as this case can easily be detected in~$O(n^3 \cdot m)$ time~\cite{FordFulkerson56} and the answer is then always yes.
This implies that~$\deg_G(s),\deg_G(t) > \beta$.
Furthermore, by \Cref{lem:borders} we can assume that there is no vertex $v$ with $\e{v} < \s{s}$ or $\s{v} > \e{t}$.
By \Cref{lem:monotone} we can assume that we search for a solution in which for all~$v_i,v_j \in V \setminus \{s,t\}$ with~$\s{v_i} < \s{v_j}$ it holds that~$\dist(s,v_i) \leq \dist(s,v_j)$.
Hence we construct a table~$T$ which stores for each vertex~$v_i \in V \setminus \{s,t\}$ and each possible distance~$d \in [2,\lambda]$ the minimum number of edges that have to be deleted from $G-\{t\}$ to ensure that all vertices~$v_j \in V \setminus\{s,t\}$ with~$\s{v_j} \geq \s{v_i}$ have distance at least~$d$ from~$s$, and furthermore, $\dist (s, v_k) \le \dist (s, v_\ell)$ holds for all $k\le \ell \le i$.
Observe that~$\dist(s,s) = 0$ in any graph and since we are looking for a solution in which~$\dist(s,t) > \lambda$, we will search for a solution in which all neighbors~$u$ of~$t$ satisfy~$\dist(s,u) \geq \lambda$.
In a last step we will then try all neighbors of~$t$ to be the last vertex before~$t$ in a shortest~$s$-$t$-path to find an optimal solution.
To avoid confusion recall that all vertices \emph{except for~$s$ and~$t$} are labeled by~$v_1,v_2,\ldots, v_{n-2}$.
We initialize~$T$ by setting~$T[v_i,2] = |\{v_j \mid \s{v_j} \leq \e{s} \land j \geq i\}|$ for all~$v_i$ with~$\s{v_i} \leq \e{s}$ and~$T[v_{\ell},2] = 0$ for all vertices~$v_{\ell}$ that are not adjacent to~$s$ as any non-neighbor~$w$ of~$s$ has distance at least two from~$s$, and $\s{w} > \e{s}$.
We further initialize~$T[v_1,d] = \deg(s)$ for all~$d \geq 3$.
We also store for each table entry~$T[v_i,d]$ with~$d \geq 2$ in a second table~$S[v_i,d]$ the vertex~$v_j$ with maximum~$\s{v_j}$-value such that all edges~$\{v_{\ell},v_r\}$ with~$\ell < j$ and~$r \ge i$ are contained in a minimal cut guaranteeing that each vertex~$v_{r'}$ with~$r'\geq i$ has distance at least~$d$ from~$s$.
We initialize~$S[v_i,2] = 1$ for all~$v_i$ and~$S[v_1,d] = 1$ for all~$d > 2$ as we only delete edges incident to~$s$ in these cases.
For increasing values of~$d$, we iterate over all vertices~$v_i \in V \setminus \{s,t,v_1\}$ in order of~$\s{v_i}$ and compute
\begin{align*}
T[v_i,d] = \min_{j \le i} \{T[v_j,d-1] + C[S[v_j,d-1],v_j,v_i]\}, \text{ and}\\
S[v_i,d] = \argmin\limits_{j \le i} \{T[v_j,d-1] + C[S[v_j,d-1],v_j,v_i]\},
\end{align*}
where~$C[v_h,v_i,v_j]$ is a function that represents for each triple of vertices~$(v_h,v_i, v_j)$ with~$h < i < j$ the size of a minimal cut (the number of edges to delete from~$G$) to ensure that there is no edge between a vertex~$v_{\ell}$ with~$h \leq \ell < i$ and a vertex~$v_r$ with~$r \geq j$.
For technical reasons we exclude~$s$ here and hence the formal definition is~$C[v_h,v_i,v_j] = |\{\{v_{\ell},v_r\}\in E \mid h \leq \ell < i \land r \geq j \land v_{\ell} \neq s \neq v_r\}|$.
The vertex~$v_h$ will only be used to avoid double counting.

		We will continue by proving that~$T$ computes exactly what it is supposed to and end this section by showing how to compute the solution for \lbc{} on proper interval graphs using this table and analyzing the running time.
Since the computation of~$S$ after initialization is trivial, we will focus on the computation of~$T$.
Assume towards a contradiction that there is a vertex~$v_i$ and a distance~$d \geq 2$ such that~$T[v_i,d]$ does not contain the minimal cost to make all vertices~$v_j$ with~$j \geq i$ have distance at least~$d$ from~$s$.
Then there is also a smallest~$d$ such that there is a vertex~$v_i$ for which~$T[v_i,d]$ is computed wrongly and we assume that~$v_i$ is the vertex with the smallest index such that~$T[v_i,d]$ is computed wrongly.
There are two cases: Either~$d=2$ or~$d>2$.
For~$d=2$ observe that every vertex that is not adjacent to~$s$ has distance at least~$2$ from~$s$ (except for~$s$).
Hence~$T[v_i,2] = 0$ is correct if~$\s{v_i} > \e{s}$, that is,~$v_i$ and~$s$ are not adjacent (recall that by \Cref{lem:borders} all vertices completely ``left'' of~$s$ are deleted).
For all vertices~$v_i$ that are adjacent to~$s$, we have to count the number of edges between~$s$ and vertices ``right'' of~$v_i$, that is, the number of vertices~$v_j$ with $\s{v_j} \leq \e{s}$ and~$j \geq i$.
Since we compute this, we can assume that~$d>2$.
Note that also~$S[v_i,2]$ is computed correctly as we only consider edges incident to~$s$ in the respective computation of~$T[v_i,2]$.
For~$d>2$ we distinguish two cases:
Either~$T[v_i,d]$ contains the size of a cut that is too large or the size of a cut that is too small and hence does not fulfill all requirements.

		If~$T[v_i,d]$ is too small, then this means that there is no cut of size~$T[v_i,d]$ such that (i) all vertices right of~$v_i$ (including~$v_i$) have distance at least~$d$ from~$s$ and (ii) that $\dist (s, v_j) \le \dist (s, v_\ell) $ for all $j\le \ell \le i$.
Let~$v_j$ be the vertex with~$j \leq i$ minimizing~$T[v_j,d-1] + C[S[v_j,d-1],v_j,v_i]$.
Since we assume that~$T[v_j,d-1]$ is computed correctly (recall that~$d$ was chosen to be the minimal value for which~$T$ was wrongly computed), we know that there is a set~$F_1$ of~$T[v_j,d-1]$ edges such that (i) the distance of all vertices~$v_r$ with~$r \geq j$ to~$s$ is at least~$d-1$, (ii) $\dist (s, v_\ell) \le \dist (s, k)$ for $\ell \le k \le j$, and (iii) that all edges~$\{v_{\ell},v_r\}$ with~$\ell < S[v_j,d-1]$ and~$r \geq j$ are contained in~$F_1$.
Since~$C[S[v_j,d-1],v_j,v_i]$ represents the cost to remove all edges between vertices~$v_{\ell'}$ with~$S[v_j,d-1] \leq \ell' < j$ to vertices~$v_r$ with~$r \geq i$, we know that in the graph $H$ arising from $G$ through the removal of both~$F_1$ and these edges, there is no edge between a vertex of distance at most~$d-2$ from~$s$ to a vertex~$v_r$ with~$r \geq i$.
Hence each such vertex~$v_r$ is of distance at least~$d$ from~$s$.
It remains to show that we have $\dist_H (s, v_\ell ) \le \dist_H (s, v_k) $ for $\ell \le k\le i$.
Since $T[v_j, d-1]$ is minimal, it contains no edge $\{v_p, v_q\} $ with $p, q \ge j$.
Furthermore, if it contains an edge $\{v_p, v_q\}$ with~$p < j$ and~$q>j$, then it contains this edge for all $p'<p$ with $\{v_{p'}, v_q\}\in E$.
Thus, no such $\ell$ and $k$ can exist.
Hence~$T[v_i,d] = |F_1| + C[S[v_j,d-1],v_j,v_i]$ is not too small.

		If~$T[v_i,d]$ is too large, then this means that there is a minimum cut~$F'$ that contains less than~$T[v_i,d]$ edges such that all vertices right of~$v_i$ (including~$v_i$) have distance at least~$d$ from~$s$, and $\dist (s, v_j)\le \dist (s, v_\ell)$ for all $j\le \ell \le i$.
In the graph~$H = (V - \{t\},E\setminus (F'\cup \{\{v,t\}\mid v \in V\}))$, there is a vertex~$v_j$ such that~$v_j$ and all vertices~$v_r$ with~$r \geq j$ have distance at least~$d-1$ from~$s$ and all vertices~$v_{\ell}$ with~$\ell < j$ have distance at most~$d-2$ from~$s$.
Hence $F'$ has to contain all edges in~$F'' = \{\{v_{\ell},v_r\} \in E \mid \ell \leq j < i \leq r\}$ as any remaining edge would yield that there is some vertex~$v_r$ with~$r\geq i$ which has distance at most~$d-1$ from~$s$ and thus~$F'$ is not a cut with the desired properties.
We partition this set~$F''$ into two disjoint sets~$F'_{\ell} = \{\{v_{\ell},v_r\} \in E \mid \s{v_\ell} < \s{S[v_j,d-1]} \land \s{v_i} \leq \s{v_r}\}$ and~$F'_r = \{\{v_{\ell},v_r\} \in E \mid S[v_j,d-1] \leq \s{v_\ell} \leq \s{v_j} \land \s{v_i} \leq \s{v_r}\}$.
Let~$F^*:= F' \setminus F'_{\ell}$ and define $H^*\coloneqq H + F'_{\ell}$.
Notice that $H^*$ fulfils that $\dist_{H^*} (s, v_\ell ) \ge d - 1$ for all $\ell \ge j$ and $\dist_{H^*} (s, v_\ell) \le \dist_{H^*} (s, v_k) $ for all~$\ell \le k$ as only edges between vertices $v$ with $\dist_H (s, v) \ge d-1$ have been deleted from $H$.
Thus, we get that $T[v_j, d-1] \le |F^*| = |F'| - |F_\ell'|$, as $T[v_j, d-1]$ was computed correctly by the definition of~$v_i$ and~$d$.
Note further that~$|F'_r|$ is by definition equal to~$C[S[v_j,d-1],v_i,v_j]$.
Thus, it holds that~$|F'| \geq T[v_j,d-1] + C[S[v_j,d-1],v_j,v_i] \geq  \min_{j' \leq i} \{T[v_j',d-1] + C[S[v_j',d-1],v_j',v_i]\} = T[v_i,d]$.

		It remains to discuss how to find the solution using~$T$ and the running time needed to compute~$C,T,$ and the solution.
Observe that it is only required that~$t$ has distance at least~$\lambda+1$ from~$s$ and not necessarily all other vertices~$v_j$ with~$\s{v_j} > \s{t}$.
Thus, for each $i$ we can construct a solution of size $T[v_i,\lambda] + |\{v_{\ell} \mid \ell < i \land \{v_{\ell},t\}\in E\}|\}$.
To see that there is at least one optimal solution of this form, we apply \Cref{lem:monotone} to get a solution $F$ with $\dist_{G-F} (s, v_i) \le \dist_{G-F} (s, v_j) $ for $i \le j$.
We can clearly assume that there is some $i$ such that $\{\{v_j, t\} \in E \mid j<i\} = F \cap \{\{v,t\}\mid v \in V\}$.
Furthermore, $H:= (V\setminus \{t\}, E\setminus (F \cup \{\{v,t\}\mid v \in V\}))$ fulfills that $\dist_H (v_j) \ge d$ for all $j\ge i$, and that $\dist_H (s, v_j) \le \dist_H (s, v_k)$ for all $j\le k$, showing that $|F| \le T[v_i,\lambda] + |\{v_{\ell} \mid \ell < i \land \{v_{\ell},t\}\in E\}|\}$.
This leads us to the conclusion that the minimum cost to make the distance between~$s$ and~$t$ at least~$\lambda+1$ is~$\min_{i} \{T[v_i,\lambda] + |\{v_{\ell} \mid \ell < i \land \{v_{\ell},t\}\in E\}|\}$.

		We conclude with the running time.
Observe that~$C[v_h,v_i,v_j]$ can be computed in~$O(m)$ time by simply iterating over all edges and checking whether the two endpoints fulfill the requirement that one endpoint is between~$h$ and~$i$ and the other is right of~$j$.
Since we precompute this function for~$O(n^3)$ pairs of vertices, the overall running time is~$O(n^3 \cdot m)$.
The table entry~$T[v_i,d]$ can be computed in~$O(n)$ time by iterating over at most~$n$ intervals and computing the sum of two table entries (one from~$T$ and one from~$C$).
Since there are~$O(n \cdot \lambda)$ table entries, the overall running time is~$O(n^2 \cdot \lambda)$.
As we may assume that~$\lambda < n$ (each path has length at most~$n$), the running time is upper bounded by~$O(n^3)$.
Lastly, computing the solution takes~$O(n^2)$ time as we have to iterate over~$n$ vertices~$v_i$ and for each we have to compute~$|\{v_{\ell} \mid \ell < i \land \{v_{\ell},t\}\in E\}|$.
This computation takes~$O(n)$ time as we only have to iterate over all neighbors of~$t$.
Hence the overall running time for our algorithm is~$O(n^3 \cdot m)$.
\end{proof}

\section{Conclusion}
In this paper we studied \lbc{} with respect to feedback vertex number, the combined parameter pathwidth plus maximum degree and the special case when the input graph is a proper interval graph.
We showed that it is W[1]-hard with respect to feedback vertex number and polynomial-time solvable on proper interval graphs.
The latter proves an open conjecture by Bazgan et al.~\cite{BazganFNNS19} and both fill-in gaps in their hierarchies for \lbc{} from a parameterized respectively graph-classes point of view.
Natural next steps include the remaining open questions in these hierarchies, in particular, interval graphs are the last remaining graph class in their graph-class hierarchy for \lbc.
We conjecture that it should be possible to extend our \Cref{thm:lcutPolyPINT} to also work on interval graphs.
Lastly, we showed that \lbc{} is W[1]-hard with respect to the combined parameter pathwidth and maximum degree.
This combines two results by Dvořák and Knop~\cite{DvorakK18} and Bazgan et al.~\cite{BazganFNNS19}. It strengthens the former, which states that the problem is W[1]-hard with respect to the parameter pathwidth, and complements the latter, which shows that the problem is in XP for the parameter maximum degree. The question whether it is FPT or W[1]-hard for the parameter maximum degree was left open by Bazgan et al.~\cite{BazganFNNS19}.

\paragraph*{Acknowledgement}
Klaus Heeger was supported by DFG Research Training Group 2434 ``Facets of Complexity''.
Du\v{s}an Knop was partially supported by the DFG under project ``MaMu'', NI 369/19 and by the OP VVV MEYS funded project CZ.02.1.01/0.0/0.0/16\_019/0000765 ``Research Center for Informatics''.

\bibliographystyle{plain}
\bibliography{literature}
\end{document}